\documentclass{article} %
\usepackage[T1]{fontenc}
\usepackage[margin=1in]{geometry}
\usepackage{amsfonts,amsmath,amsthm,amssymb,mathtools}  %

\mathtoolsset{centercolon}
\usepackage{xfrac,nicefrac}
\usepackage{mathdots}
\usepackage{mleftright}  %
\let\left\mleft
\let\right\mright

\usepackage{xspace}
\xspaceaddexceptions{]\}}  %
\usepackage{regexpatch}

\usepackage{bm,bbm,dsfont}  %
\usepackage{caption}
\usepackage[normalem]{ulem}
\usepackage{enumitem}

\usepackage{graphicx}
\usepackage{float}
\usepackage{subcaption}  %
\usepackage{tcolorbox}
\usepackage{tikz}
\usetikzlibrary{decorations.pathreplacing}
\usetikzlibrary{calc}
\usetikzlibrary{positioning}
\usetikzlibrary{arrows.meta}

\usepackage[linesnumbered,boxed,ruled,vlined]{algorithm2e}

\SetCommentSty{mycommfont}

\usepackage{thmtools,thm-restate}
\usepackage[colorlinks,citecolor=blue,linkcolor=blue,urlcolor=red]{hyperref}

\theoremstyle{plain}

\newtheorem{theorem}{Theorem}[section]  %
\newtheorem{lemma}[theorem]{Lemma}

\newtheorem{claim}[theorem]{Claim}
\theoremstyle{definition}  %

\newtheorem{remark}[theorem]{Remark}

\usepackage[capitalise]{cleveref}
\crefname{algocf}{Algorithm}{Algorithms}
\Crefname{algocf}{Algorithm}{Algorithms}
\crefname{claim}{Claim}{Claims}
\Crefname{claim}{Claim}{Claims}

\DeclarePairedDelimiter{\bk}{(}{)}
\DeclarePairedDelimiter{\Bk}{[}{]}
\DeclarePairedDelimiter{\BK}{\{}{\}}

\DeclarePairedDelimiter{\abs}{\lvert}{\rvert}

\DeclareMathOperator*{\E}{\mathbb{E}}
\DeclareMathOperator*{\Var}{Var}

\let\Pr\PrAux
\DeclareMathOperator{\poly}{poly}

\DeclareMathOperator*{\ind}{\mathbbm{1}}

\renewcommand{\tilde}{\widetilde}

\newcommand{\defeq}{\coloneqq}
\newcommand{\eps}{\varepsilon}

\renewcommand{\epsilon}{\eps}

\usepackage{regexpatch}
\makeatletter
\xpatchcmd\thmt@restatable{%
\csname #2\@xa\endcsname\ifx\@nx#1\@nx\else[{#1}]\fi
}{%
\ifthmt@thisistheone
\csname #2\@xa\endcsname\ifx\@nx#1\@nx\else[{#1}]\fi
\else
\csname #2\@xa\endcsname[{Restated}]
\fi}{}{}
\makeatother

\newcommand{\hashto}{\xmapsto{h}}

\newcommand{\Index}{\texttt{index}}

\title{Optimal Non-Oblivious Open Addressing}
\author{
Michael A. Bender%
\thanks{Partially supported by NSF grants CCF 2247577 and  CCF 2106827 and the John L.~Hennessy Chaired Professorship. \texttt{bender@cs.stonybrook.edu}.}\\
Stony Brook University
\and
William Kuszmaul%
\thanks{Partially supported by a Harvard Rabin Postdoctoral Fellowship and by a Harvard FODSI fellowship under NSF grant DMS-2023528. \texttt{kuszmaul@cmu.edu}.}\\
CMU
\and
Renfei Zhou%
\thanks{Partially supported by a MongoDB PhD Fellowship. \texttt{renfeiz@andrew.cmu.edu}.}\\
CMU
}
\date{}

\begin{document}

\maketitle

\begin{abstract}
A hash table is said to be \emph{open-addressed} (or \emph{non-obliviously open-addressed}) if it stores elements (and free slots) in an array with no additional metadata. Intuitively, open-addressed hash tables must incur a space-time tradeoff: The higher the load factor at which the hash table operates, the longer insertions/deletions/queries should take.

In this paper, we show that no such tradeoff exists: It is possible to construct an open-addressed hash table that supports constant-time operations \emph{even when the hash table is entirely full}. In fact, it is even possible to construct a version of this data structure that: 

\begin{itemize}
\item is dynamically resized so that the number of slots in memory that it uses, at any given moment, is the same as the number of elements it contains;
\item supports $O(1)$-time operations, not just in expectation, but with high probability;
\item requires external access to just $O(1)$ hash functions that are each just $O(1)$-wise independent.
\end{itemize}

Our results complement a recent lower bound by Bender, Kuszmaul, and Zhou showing that \emph{oblivious} open-addressed hash tables must incur $\Omega(\log \log \epsilon^{-1})$-time operations. The hash tables in this paper are non-oblivious, which is why they are able to bypass the previous lower bound.
\end{abstract}

\thispagestyle{empty}
\setcounter{page}{0}
\newpage

\section{Introduction}
\label{sec:intro}

Starting with the introduction of linear-probing hash tables in 1954 \cite{Knuth98Vol3}, research on \emph{open-addressed hash tables} has persisted for more than seven decades as one of the most widely studied and fruitful subareas of randomized data structures. In its most general form, an open-addressed hash table is defined by its state: At any given moment, the hash table is simply an array storing elements and free slots in some order.\footnote{As we will discuss shortly, this class of data structures is sometimes also referred to as \emph{non-oblivious} open addressing \cite{fiat1988nonoblivious, fiat1993implicit} to distinguish from several subclasses (e.g., greedy open addressing \cite{yao1985uniform, ullman1972note} and oblivious open addressing \cite{munro1986techniques, oblivious}) that have also been referred to in past works simply as \emph{open addressing}.} If $k$ elements are stored in an array of size $n$, then the hash table is said to support \emph{load factor} $k / n$.

The main thrust of research on open-addressed hash tables has been towards answering the following basic question: Given the restriction that the hash table is simply an array of elements and free slots, and given a load-factor $(1 - \epsilon)$, what are the best possible time bounds that one can achieve for insertions/deletions/queries?

Early work on (ordered) linear-probing hash tables led to bounds of the form $O(\epsilon^{-2})$ for insertions/deletions and $O(\epsilon^{-1})$ for queries \cite{knuth1963notes, amble1974ordered, celis1985robin, konheim1966occupancy}; the query time was subsequently improved to $O(\log \log \epsilon^{-1})$ using interpolation techniques for searching within each run\footnote{A \emph{run} in an ordered linear-probing hash table is a maximal consecutive interval of non-empty slots.} in the table \cite{gonnetinterpolation}.  Bucketized cuckoo hashing \cite{dietzfelbinger2007balanced} can be used to achieve expected insertion/deletion time $O(2^{\poly \log \epsilon^{-1}})$ while placing items in one of two buckets of size $O(\log \epsilon^{-1})$ -- using interpolation techniques, this leads to $O(\log \log \log \epsilon^{-1})$-time queries. In fact, using techniques developed by Brent in the early 1970s \cite{brent1973reducing}, and then refined by later sets of authors \cite{gonnet1979efficient, mallach1977scatter}, it is known how to achieve positive queries in $O(1)$ expected time while supporting $\poly(\epsilon^{-1})$-expected-time insertions/deletions. Or, if one is willing to abandon insertions/deletions (i.e., they take $\poly n$ time), it is even known how to construct a static hash table with load factor $1$ that supports all queries in $O(1)$ worst-case time with high probability \cite{fiat1988nonoblivious, fiat1993implicit}. These are just a few samples from a rather large design space that has been explored over half a century.

All of the time bounds above encounter a seemingly natural barrier: At least \emph{one} of the operations (either queries, insertions, or deletions) takes $\Omega(\epsilon^{-1})$ time. Only in very recent work by Bender, Kuszmaul, and Zhou \cite{oblivious} has this barrier finally been overcome. Bender et al.~\cite{oblivious} achieve $O(1)$-expected-time queries and $O(\log \log \epsilon^{-1})$-time insertions/deletions. 

An interesting feature of the hash table in \cite{oblivious} is that queries are \emph{oblivious}. That is, the sequence of positions examined by the query algorithm for an element $u$ is given by a fixed \emph{probe sequence} $h_1(u), h_2(u), \ldots \in [n]$. There is also a matching lower bound for this class of hash tables: Any oblivious solution must incur $\Omega(\log \log \epsilon^{-1})$ amortized expected time for at least one of its operations \cite{oblivious}. This leads~\cite{oblivious} to conclude with the following open question: Is it possible to achieve better bounds using \emph{non-oblivious} solutions? Or, are the oblivious bounds optimal across all open-addressed hash tables?

\paragraph{Partner hashing: A constant-time solution with load factor 1.} This paper introduces \emph{partner hashing}, a non-oblivious open-addressed hash table with the following set of guarantees:
\begin{itemize}
\item \textbf{Load factor 1:} At any given moment, if there are $n$ elements, then the hash table is an array of size $n$ storing the elements in some order. There are no free slots and there is no external metadata. 
\item \textbf{Dynamic resizing:} The above guarantee remains true even as $n$ changes over time. The hash table can be viewed as occupying a prefix of memory that, on insertions, increases its size by one, and on deletions, decreases its size by one.
\item \textbf{Constant-time operations: }Each operation (insertions, deletions, and queries) takes constant time, not just in expectation, but with probability $1 - 1 / \poly(n)$, where $n$ is the current size of the hash table. 
\item \textbf{\boldmath$O(1)$-independent hash functions:} The hash table requires external access to just $O(1)$ hash functions that are each just $O(1)$-wise independent.
\end{itemize}

Our results come with a surprising takeaway: For non-oblivious open-addressed hash tables, there is no fundamental tension between $\epsilon$, the expected running times of operations, or even the high-probability worst-case running time of operations. Even the act of performing dynamic resizing, which might seem at first glance to require frequent rebuilds of the entire data structure, does not incur any inherent dependency on $\epsilon$.

At a technical level, the main contribution of partner hashing is a remarkably efficient approach to encoding meaningful metadata in the relative ordering of elements. When $n$ distinct elements are stored in an array of size $n$, there are $\log n! = \Theta(n \log n)$ bits of information that, in principle, could be encoded by the order of the elements. The problem is that, a priori, it would seem difficult to make use of these bits without compromising at least one type of operation (either queries or insertions/deletions). Past work on static non-oblivious open-addressed hash tables \cite{fiat1988nonoblivious, fiat1993implicit} showed that, if one is willing to rebuild the entire hash table on each operation, then it is possible to encode $\Theta(n \log n)$ bits of information in the data structure state, so that queries can retrieve words of size $\Theta(\log n)$ bits in $O(1)$ time -- in other words, the data structure can simulate a random-access memory (RAM) of $\Theta(n \log n)$ bits with words of size $O(\log n)$ bits and $O(1)$-time accesses. However, these constructions are tied intrinsically to the static case, as they require that the entire hash table forms a single, elaborate permutation, whose state may change entirely if a single RAM update were to be performed.

Partner hashing introduces a (completely) new technique for encoding a $\Theta(n \log n)$-bit RAM, again with $O(\log n)$-bit RAM words, and again where the entire RAM is encoded in the ordering of elements. The difference is that, now, \emph{both} updates and queries to the RAM can be supported in $O(1)$ time -- and all without interfering with the normal operations of the hash table, and even without requiring that the hash table has a static size. This allows us to make use of the natural redundancy that any open-addressed hash table comes with (by virtue of encoding a permutation of $n$ elements), without requiring the use of free slots at all. The fact that such a RAM can be implemented in an open-addressed hash table (and with high-probability time bounds, no less!) came as quite a surprise to the current authors, and was only discovered in the course of attempting to prove an impossibility result for the same problem.

\subsection{Other related work}

Since the introduction of hash tables in 1953 \cite{knuth1963notes}, the study of hash tables has grown into one of the largest sub-areas of data structures. Even very simple hash table designs have been the subject of a great deal of theoretical interest -- this includes work on linear probing \cite{knuth1963notes, amble1974ordered, konheim1966occupancy, pagh2009linear, patrascu2016independence, bender2022linear}, double hashing \cite{lueker1988more, molodowitch1990analysis,burkhard2005external,martini2003double,GuibasSz78, mitzenmacher2012peeling}, quadratic probing \cite{HopgoodDa72, Maurer68, ecker1974period, batagelj1975quadratic, radke1970use, kuszmaul2024towards}, uniform probing \cite{ullman1972note, yao1985uniform, greedy}, and cuckoo hashing \cite{pagh2001cuckoo, dietzfelbinger2007balanced, frieze2018balanced, bell20241, arbitman2009amortized, kirsch2010more, kirsch2007using, DietzfelbingerGoMi10, fountoulakis2016multiple, frieze2019insertion,frieze2018balanced,fountoulakis2013insertion, frieze2011analysis, khosla2019faster, fountoulakis2010orientability, frieze2012maximum,walzer2022insertion, cain2007random, fernholz2007k, lelarge2012new,lehman20093, walzer2023load,kuszmaul2025efficient}. See, also \cite{munro1986techniques}, for an early survey on the topic. Many of these hash-table designs are still the subject of widely-studied open questions, including the questions of determining the space-time tradeoffs that are achieved by quadratic probing \cite{kuszmaul2024towards} and variations of cuckoo hashing \cite{frieze2018balanced, dietzfelbinger2007balanced}.

Within open-addressed hash tables, one of the major pushes has been to understand the optimal time/space bounds that a hash table can hope to achieve. Past works have determined the optimal tradeoffs for sub-classes of open addressing, including oblivious open addressing \cite{oblivious}, greedy open addressing \cite{ullman1972note, yao1985uniform, greedy}, incremental open addressing without reordering \cite{greedy}, history-independent open addressing \cite{DBLP:journals/iacr/NaorT01, DBLP:conf/focs/BlellochG07, kuszmaul2023strongly}, and static open addressing \cite{fiat1988nonoblivious,fiat1993implicit, brent1973reducing, gonnet1979efficient, mallach1977scatter}. The largest class of open-addressed hash tables is the one that we study in this paper, non-oblivious open addressing. Our results reveal that, somewhat surprisingly, this class does not exhibit \emph{any} non-trivial tradeoff between time and space.

There has also been a great deal of work outside of the open-addressing model. This includes long lines of work on succinct hash tables \cite{bender2022optimal, li2023dynamic, li2024dynamic, li2023tight, bender2023iceberg, bercea2023dynamic, Raman03Succinct, ArbitmanNaSe10}, external-memory hash tables \cite{IaconoPa12, ConwayFaSh18, JensenPa08, verbin2013limits}, and hash tables with super-high probability or even deterministic guarantees \cite{GoodrichHiMi12, kuszmaul2022hash, bender2023iceberg, Sundar91,HagerupMiPa01, Ruzic08,PatrascuTh14}. Succinct hash tables, in particular, are of interest, because by storing non-explicit representations of keys, they are able to use space close to the information-theoretic optimum. For an overview of some of the major techniques used in succinct hashing, see, also \cite{bender2024modern}.

Closely related to the theory of hash tables is the theory of how to simulate random \emph{hash functions} time and space efficiently. Starting with Carter and Wegman \cite{wegman1979new, carter1977universal}, there has been a long line of work \cite{siegel1989universal, siegel1995universal, siegel2004universal, ostlin2003uniform, pagh2001cuckoo, dietzfelbinger2009applications} on how to simulate high-levels of independence without blowing up time or space. Perhaps the most notable of these results is the construction by Siegel \cite{siegel1995universal} that supports $n^{\epsilon}$-independence with $n^{O(\epsilon)}$ space and constant evaluation time. There is also interest in how much independence is required for specific hash tables, such as linear probing \cite{pagh2009linear, patrascu2016independence} (for which 5-independence is known to be optimal) and cuckoo hashing \cite{cohen2009bounds} (the optimal independence for which is still open).

Open-addressed hash tables are special cases of \emph{implicit data structures}, whose defining characteristic is that the entire data structure is simply an array of elements, and the structural information is encoded implicitly in the ordering of elements, rather than explicitly using pointers. There has been a series of works studying implicit \emph{ordered} dictionaries \cite{munro1980implicit,frederickson1983implicit,munro1984implicita,munro1986implicit,franceschini2002implicita,franceschini2003implicit,franceschini2003optimal,franceschini2006optimal}, culminating in a solution due to Franceschini and  Grossi \cite{franceschini2006optimal} that achieves $O(\log n)$ amortized time for a variety of operations including insertions, deletions, and predecessor searches. For implicit \emph{unordered} dictionaries, Yao \cite{yao1981should} shows that deterministic data structures over a sufficiently large universe must have $\Omega(\log n)$ query time even in the static setting. In contrast, when the universe size is subexponential in $\poly n$, Fiat and Naor \cite{fiat1993implicit} show that static solutions can acheive deterministic query time $O(1)$. Additionally, for arbitrary universe sizes $m$, it is known how to construct randomized static solutions \cite{fiat1988nonoblivious, fiat1993implicit} that achieve $O(1)$ worst-case query time with high probability, and deterministic solutions \cite{fiat1988nonoblivious, fiat1993implicit} that achieve constant time with access to $O(\log \log m)$ bits of external storage.

Several techniques and results from past works will be directly useful in our constructions. To encode metadata, we make use of the space-efficient retreival data structures developed by Demaine, Meyer auf der Heide, Pagh, and P{\v a}tra{\c s}cu~\cite{demaine2006dictionariis} (since we don't get to store metadata explicitly, this data structure will be stored implicitly in the ordering of elements). To turn expected insertion/deletion time guarantees into high-probability guarantees we will use (along with some additional ideas) the idea of buffering work so that it can be spread out over multiple operations, a technique that, in the context of hash tables, was first developed first by Dietzfelbinger and Meyer auf der Heide~\cite{dietzfelbinger1990new}, and that has since been applied to a variety of other constant-time hash tables~\cite{ArbitmanNaSe10, arbitman2009amortized,bender2024modern}. We will also draw on many of the basic design paradigms that are used within succinct hashing, including the idea hashing keys to random buckets of size $\poly \log n$ \cite{bender2022optimal, li2024dynamic, bender2023iceberg, bercea2023dynamic, ArbitmanNaSe10}, and the idea of mapping each key to its position within the bucket using a small (and in our case implicit) secondary data structure \cite{bender2022optimal, li2024dynamic, bender2023iceberg, bercea2023dynamic, ArbitmanNaSe10}.

\section{Preliminaries}
\label{sec:prelim}

\paragraph{Open addressing. }
An \emph{unordered dictionary} is a data structure that stores a set of keys from some universe $U$, while supporting insertions (which add a new key), deletions (which remove a key), and membership queries (which check if a given key is present).\footnote{In most such dictionaries -- including all open-addressed hash tables, which are the topic of this paper -- the dictionary can also be augmented to store a value associated with each key.} As a cultural convention, an unordered dictionary is referred to as a \emph{hash table} if it uses hash functions in its design.

A hash table is said to be \emph{open-addressed} if 
its entire state is simply an array $A[0 \,..\, n\!-\!1]$ of elements and free slots, with each element in the hash table appearing exactly once in the array. Besides this array, the hash table also gets access to three other pieces of information that are not considered to be part of the state: the array size $n$, the number of elements that are currently present, and the random bits (or, equivalently, hash functions) that the hash table requires access to. Although in our warmup construction (Section \ref{sec:simple}) we will assume access to fully random hash functions (or, equivalently, to $\Theta(U \log U)$ random bits), our final constructions (Sections \ref{sec:fixed_size} and \ref{sec:resizing}) will require access only to a constant number of $O(1)$-wise independent hash functions (i.e., to $O(\log U)$ random bits).

It should be emphasized that the array $A$ cannot be used to store arbitrary content. Rather, if the array has $n$ slots, and if the hash table is storing $k$ elements, then each of the elements must appear in exactly one slot, and the remaining $n - k$ slots must be empty. In the final constructions of this paper, we will have $n = k$, which means that the array $A$ is simply a \emph{permutation of the $n$ elements being stored}. The hash table does not get any additional space to store metadata. 

An open-addressed hash table is said to be \emph{dynamically resized} if, as insertions and deletions are performed, the hash table changes the size $n$ of the array that it uses over time. In this case, we can think of the array $A$ as representing, at any given moment, a size-$n$ prefix of an infinite tape. If a hash table is not dynamically-resized, it is said to be \emph{fixed-capacity}.

When discussing space efficiency, one typically defines the \emph{load factor} of an open-addressed hash table to be the quantity $k / n$, where $k$ is the number of elements currently stored. In our final construction (Section \ref{sec:resizing}), we will build a dynamically-resized hash table that operates continually at load factor 1. 

\paragraph{Supporting high-probability guarantees.}
A data structure is said to offer a guarantee $G$ \emph{with high probability in $n$} if, for every positive constant $c > 0$, if we assume that the constants used in the description of the data structure are chosen appropriately (possibly as a function of $c$), the event $G$ holds with probability at least $1 - O(1 / n^c)$. In this case, we also say that the event $G$ \emph{fails} with probability $1 / \poly(n)$.

Throughout the paper, if a hash table detects that a $1 / \poly(n)$-probability failure event has occurred, we will allow the hash table to sample new hash functions with which to rebuild itself. 

\paragraph{Reducing to polynomial universe size.} When constructing a fixed-capacity open-addressed hash table, one can assume without loss of generality that the universe $U$ has size at most $|U| \le \poly(n)$. In particular, one can use a pairwise independent hash function $g:U \rightarrow U'$ to map keys into a polynomial-size universe $U'$. The hash table can then act as though it is storing elements from $U'$ (i.e., it treats the array slot containing $x$ as containing $g(x)$). If, by unfortunate luck, a key $x$ being inserted has a collision $g(x) = g(y)$ for some key $y$ already present, then the hash table incurs a failure and is rebuilt. By setting $|U'| = \poly(n)$, these types of failures occur with probability at most $1 / \poly(n)$. 

This reduction does not immediately apply to dynamically-resized hash tables because $n$ (and therefore $|U'|$) changes over time. However, as we shall see in Section \ref{sec:resizing}, because of the way that our dynamically-resized solution builds on our fixed-size solution, it will be possible to maintain the reduction to a polynomial-size universe case even when we are performing resizing. (See Remark \ref{rem:universe}.)

\paragraph{Other conventions.}
We will use the notation $[a] = \BK{0, 1, \ldots, a\!-\!1}$, $[a, b] = \{a, a + 1, \ldots, b\}$, $[a, b) = \{a, a + 1, \ldots, b - 1\}$, etc. Additionally, by a minor abuse of notation, when discussing a positive constant $\delta$, we will sometimes use the notation $\delta \cdot O(f(n))$ to denote a function of the form $\delta \cdot g(n)$, where $g(n) = O(f(n))$, and where the hidden constant in the big-O notation is independent of $\delta$.

\section{\boldmath{}A Warmup Data Structure: Load Factor \texorpdfstring{$1 - o(1)$}{1 - o(1)} with Expected Operation Time \texorpdfstring{$O(1)$}{O(1)}}
\label{sec:simple}
As a warmup data structure, to demonstrate the basic techniques that we will use in the paper, we begin by proving the following theorem:

\begin{restatable}{theorem}{ThmSimple}
  \label{thm:simple}
  Assuming access to a constant number of fully random hash functions, there is an open-addressing hash table that maintains $n \bk[\big]{1 - \frac{1}{\log^c n}} \pm O(1)$ keys in $n$ slots (for any $c > 100$), supports insertions and deletions in $O(1)$ expected time, and supports queries in $O(1)$ worst-case time. The hash table requires $O(n \log n)$ bits of scratch space to initialize.
\end{restatable}

Already, the data structure in Theorem \ref{thm:simple} achieves an interesting guarantee: it shows that one can operate continually at load factor $1 - o(1)$ without compromising expected insertion/deletion/query time. 

It is worth noting, however, that this data structure is still far from achieving our final guarantees. In particular, the weaknesses of this data structure are that: it requires fully random hash functions, rather than $O(1)$-independent ones; it does not support dynamic resizing, instead requiring that the number of elements stays in a very narrow band of the form $n \bk[\big]{1 - \frac{1}{\log^c n}} \pm O(1)$; it requires scratch space to initialize; and it allows insertions/deletions to take constant expected time, rather than high-probability worst-case time. With additional technical ideas, we will be able to resolve all of these issues in later sections.

To simplify our discussion, throughout the section, we will assume that the universe of keys $U$ has $\poly(n)$ size. This assumption is without loss of generality, as noted in \cref{sec:prelim}.

The main challenge in proving Theorem \ref{thm:simple} is to use the relative order of elements in the hash table in order to maintain an implicit RAM of $\Theta(n \log n)$ bits with word-size $w = \Theta(\log n)$. This is the task that we begin with in \cref{sec:encode_ram}. 

\subsection{Encoding the RAM}
\label{sec:encode_ram}

Let $B = \log^{\Theta(1)} n$ be a parameter. The slots are divided into $m$ bins each of $B$ slots, where $mB = n$ is the size of the hash table, and where $m$ is a power of two. The hash table will accommodate $n \bk[\big]{1 - B^{-1/4}} \pm O(1)$ keys, so the load factor is $1 - 1/\poly\log n$ as desired.
We label the bins with numbers in $[m]$.

Slots are labeled with numbers in $[n]$. A slot is identified by two numbers: the bin $k \in [m]$ it belongs to, and its \emph{offset} $t \in [B]$ within that bin. The slot in bin $k$ with offset $t$ is slot $kB + t$ in the hash table.

Let $h$ be a fully-independent hash function from $U$ to $[m]$, i.e., it hashes every key $x$ to a bin $k$, written $x \hashto k$ or $h(x) = k$. Also let $r_0, \ldots, r_{n/4 - 1}$ be random numbers in $[m/2]$, which are obtained via another fully-independent hash function $g$ from $[n/4]$ to $[m/2]$ (that is, $r_i$ is a shorthand for $g(i)$, where $g$ is the hash function). 

In our design, the arrangement of the keys has two \emph{layouts}: the \emph{logical layout} and the \emph{physical layout}. Each layout maps the keys to slots. The physical layout is the true layout, i.e., it indicates the actual slot that each key resides in. However, before we can describe the physical layout, it will be helpful to first describe the logical layout, which is what the layout would be if we were not storing any information in the RAM that we are implementing. As we shall see, the only difference between the physical and logical layouts will be that certain pairs of keys are swapped with one another. 

To make clear which layout we are talking about at any given moment, we will refer to the \emph{logical (resp.~physical) bin} of a key $x$ to be the bin where it resides in the logical (resp.~physical) layout; the \emph{logical (resp.~physical) offset} of a key $x$ to be the offset of the key in the logical (resp.~physical) layout; and, finally, the \emph{logical (resp.~physical) address} of a key $x$ to be the slot number in which the key resides in the logical (resp.~physical) layout.

\paragraph{The logical layout.}
In the logical layout, we put every key $x$ into the bin $h(x)$ that it hashes to, i.e., the logical bin of each key $x$ is simply $h(x)$. Within a given bin, we will guarantee, at any given moment, that the occupied slots form a prefix of the offsets in the bin, and that the unoccupied slots form a suffix of the offsets in the bin. (The set of occupied slots will be the same in both the logical and physical layouts.) 

The hash table will fail if any of the bins overflow, i.e., if more than $B$ keys hash to the bin.

\begin{claim}
  \label{claim:no_overflow_underflow}
  As long as $B \ge \log^{4} n$, with probability $1 - n^{-\omega(1)}$, the number of keys hashed to every bin is in $\Bk[\big]{B (1 - 2B^{-1/4}),\, B}$. This implies that, with high probability, no bin overflows.
\end{claim}

\begin{proof}
  By symmetry, we only need to bound the probability that bin 0 overflows or underflows. Let $X_i \defeq \ind\Bk[\big]{\text{the $i$-th key $\hashto$ bin 0}}$, and let $X = \sum_{i=1}^{n(1 - B^{-1/4})} X_i$. According to the overall load factor, we know $\E[X] = B (1 - B^{-1/4})$. Applying a Chernoff bound, we know that
  \[
    \Pr\Bk[\big]{\abs[\big]{X - B (1 - B^{-1/4})} > B^{3/4}} \le \exp\bk[\big]{-\Omega(B^{1/2})} \le \exp\bk[\big]{-\Omega(\log^2 n)} = n^{-\omega(1)}.
  \]
  Applying a union bound over all $m$ bins concludes the proof.
\end{proof}

For the moment, we will assume access to a \emph{retrieval data structure} that maps every stored key $x$ to its logical offset within bin $h(x)$. Later on, we will show how to encode the retrieval data structure \emph{within} the RAM that we are implementing. But, to avoid the risk of cyclic dependencies, it is helpful for now to think of the retrieval data structure as being stored externally. Specifically, we will use the following retrieval data structure due to Demaine, Meyer auf der Heide, Pagh, and P{\v a}tra{\c s}cu \cite{demaine2006dictionariis} (restated for a $\poly(n)$-sized universe):

\begin{theorem}[\cite{demaine2006dictionariis}]
  \label{thm:retrieval}
  There is a dynamic retrieval data structure that maintains $n$ key-value pairs, where the keys come from a $\poly(n)$-sized universe $U$ and the values are $r$-bit strings; that supports retrieval queries in $O(1)$ worst-case time; that supports insertions and deletions in $O(1)$ time with high probability in $n$; and that uses $O(n \log \log n + nr)$ bits of space.
\end{theorem}

In summary, using the hash function $h$ and the retrieval data structure, we can recover the logical address of a key $x$ as follows: the logical bin of the key is given by $h(x)$, and the logical offset of the key is given by the retrieval data structure. If the key has logical bin $k$ and logical offset $t$, then its logical address is $kB + t$. We remark that, in our setting, because the retrieval data structure is storing $r = O(\log \log n)$ bit values (i.e., logical offsets in a bin of size $B \le \poly\log n$), so the total space needed for a dynamic retrieval data structure is $O(n \log \log n + nr) = O(n \log \log n)$ bits.

\paragraph{The physical layout.}
As noted earlier, the relationship between the logical and physical layouts is that, starting at the logical layout, we can get to the physical layout by swapping some pairs of keys with one another. Recall that every key $x$ is hashed to a logical bin via the hash function $h$. We now view this hash value $h(x)$ as a pointer from the key $x$ to that bin. If, in the physical representation, a key $x$ resides in its logical address, then key $x$ will point at the bin containing itself, in which case we say key $x$ is a \emph{self-loop}.

Assume keys $x$ and $y$ are two self-loops in different bins. We can choose to swap $x$ and $y$, making them a \emph{coupling pair}, in which case we say $x$ and $y$ are \emph{partners} of each other. As a special case, if $x$ is a self-loop, we say that $x$ is the partner of itself. Whenever two keys $x$ and $y$ are partners, the physical address/bin/offset of $x$ will be the logical address/bin/offset of $y$ and vice versa. 

Given access to the physical layout, the hash function $h$, and the retrieval data structure introduced above, we can recover the physical address of any key $x$ as follows. First, we compute the logical address $j$ of $x$ using $h$ and the retrieval data structure. If $x$ is in slot $j$, then $j$ is also the physical address for $x$. Otherwise, $x$ must be in a coupling pair with the key $y$ physically sitting in slot $j$. Once again using $h$ and the retrieval data structure, we can calculate the logical address $i$ of $y$. This address $i$ is guaranteed to be the physical address of $x$. 

More generally, assuming that $x = A[i]$ and $y = A[j]$ form a coupling pair, as long as we know any of the four objects $x, y, i, j$, we can determine the other three in $O(1)$ time (with high probability in $n$) by accessing the retrieval data structure, the hash function $h$, and the physical layout. See Figure \ref{fig:xyij}.

For future reference, we describe the procedure for finding the partner of a given key $x$ in \cref{alg:find_partner}.

\begin{algorithm}[ht]
  \caption{Finding the Partner of a Key}{\label{alg:find_partner}}
  \DontPrintSemicolon
  \SetKwProg{Fn}{Function}{:}{}
  \SetKwFunction{FindPartner}{FindPartner}
  \SetKwFunction{QueryRetrieval}{QueryRetrieval}
  \Fn(\tcp*[f]{Find the physical address of the partner of key $x$}){\FindPartner{$x$}} {
    $t \gets$ \QueryRetrieval{$x$} \tcp*{If $x$ not in the retrieval data structure, $t \in [B]$ can be arbitrary} 
    \Return $h(x) \cdot B + t$\;
  }
\end{algorithm}

\begin{figure}[ht]
  \centering
  \includegraphics[width = 0.5 \textwidth]{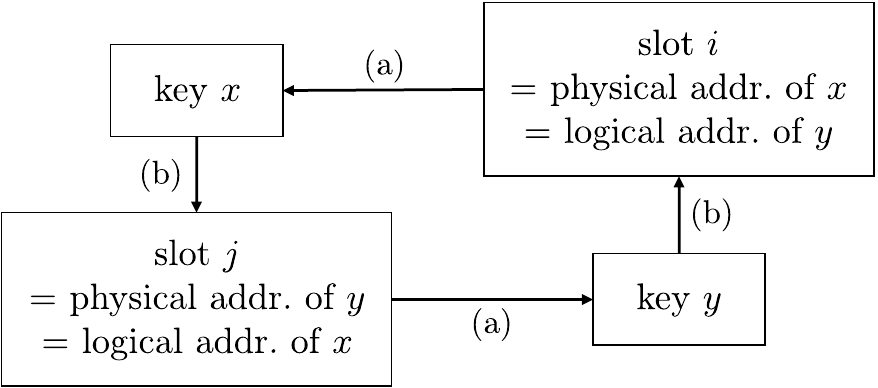}
  \caption{Using any one of $x, y, i, j$ to recover the others. (a) Given $i$ or $j$, we can recover $x$ and $y$, respectively, using the physical layout; (b) given $x$ or $y$, we can recover $j$ and $i$, respectively, using $h$ and the retrieval data structure (in particular $j$ and $i$ are the logical addresses for $x$ and $y$, respectively). Thus $i$ determines $x$, which determines $j$, which determines $y$, which determines $i$.}
  \label{fig:xyij}
\end{figure}

\begin{figure}[ht]
  \centering
  \includegraphics[width = 0.8 \textwidth]{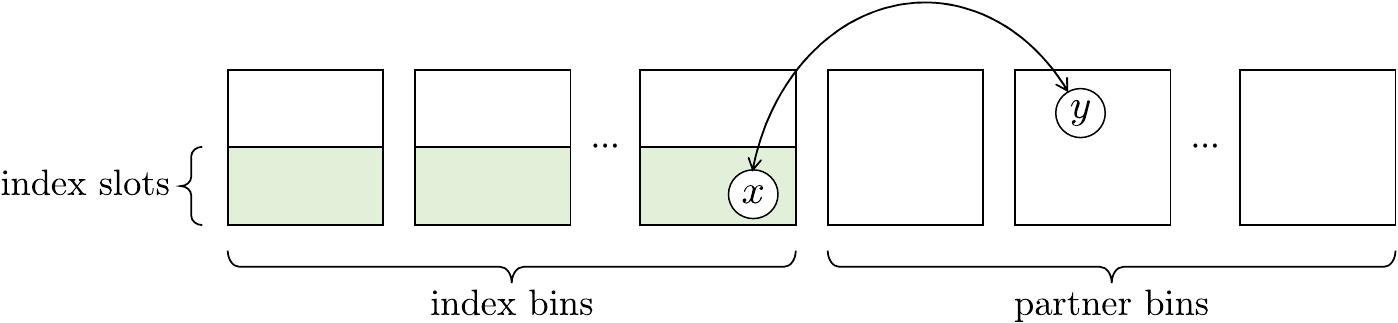}
  \caption{Index bins and partner bins. There are $m/2$ index and partner bins respectively, each containing $B = \poly \log n$ slots. To encode a value $v_i$ in the $i$-th word, we swap the key $x$ stored in slot $\Index(i)$ with an arbitrary self-loop $y$ in the $\bk[\big]{(v_i \oplus r_i) + 1}$-th partner bin.}
  \label{fig:index_and_partner_bins}
\end{figure}

\paragraph{Implementing RAM reads and writes.}
The goal of the physical layout is to encode a RAM with $\Theta(n \log n)$ bits, which will be used to store metadata for the full hash table (including the retrieval data structure of $\Theta(n \log \log n)$ bits). We divide the $m$ bins into two types: the first half of the bins are called \emph{index bins} while the second half are called \emph{partner bins}. For $i \in [m/2]$, the $(i + 1)$-th index bin has bin number $i$, while the $(i + 1)$-th partner bin has bin number $m/2 + i$. 

In every index bin, the first half of the slots (with offsets $t < B/2$) are called \emph{index slots}. This means that there are $mB/4 = n/4$ index slots in total. For $i \in [n / 4]$, we use $\Index(i) \in [n]$ to denote the position of the $(i + 1)$-th index slot in the hash table, so $A[\Index(i)]$ is the key stored in the $i$-th index slot. Note that $A[\Index(i)]$ will never be an empty slot, since by \cref{claim:no_overflow_underflow}, and because the occupied slots within each bin form a prefix of the offsets, all of the index slots are guaranteed to be occupied at any given moment.

Suppose that we wish to store values $v_0, v_1, \ldots, v_{m/4 - 1}$ in the RAM that we are implementing. We encode $v_i$ using index slot $i$. Specifically, if $x$ is the key with logical address $\Index(i)$, then we partner $x$ with some key $y$ whose logical bin is the $((v_i \oplus r_i) + 1)$-th partner bin (i.e., bin $m/2 + (v_i \oplus r_i)$). This means that $y$ has physical address $\Index(i)$. Thus, to recover $v_i$, we can compute $y = A[\Index(i)]$, and then compute $v_i = (h(y) - m/2) \oplus r_i$. The full procedure for reading a RAM entry is given by \cref{alg:ram_read}, which takes $O(1)$ time in the worst case.

\begin{algorithm}[ht]
  \caption{Reading a RAM Entry}{\label{alg:ram_read}}
  \DontPrintSemicolon
  \SetKwProg{Fn}{Function}{:}{}
  \SetKwFunction{ReadRAM}{ReadRAM}
  \Fn{\ReadRAM{$i$}} {
    $x \gets A[\Index(i)]$\;
    \uIf(\tcp*[f]{$x$ is hashed to the $(k + 1)$-th partner bin}){$k \defeq h(x) - m/2 \ge 0$} {
      \Return $k \oplus r_i$\tcp*{$r_i$ are fixed random integers in $[m/2]$}\label{line:rand_shift}
    }\Else(\tcp*[f]{$x$ is a self-loop}){
      \Return ``RAM word $i$ is empty''\;
    }
  }
\end{algorithm}

Next, we describe how to perform a write to a RAM entry (see \cref{alg:ram_write}). Suppose we wish to write value $v_i$ to RAM word $i$. First, we need to erase the old value in this RAM word by breaking the old coupling pair containing the $(i + 1)$-th index slot. Assuming constant-time access to the retrieval data structure, this step takes $O(1)$ time in the worst case. Second, we encode the new value $v_i$ by forming a coupling pair between the element $x$ whose logical address is $\Index(i)$ and some self-loop element $y$ whose logical bin is the $\bk[\big]{(v_i \oplus r_i) + 1}$-th partner bin. We find $y$ by repeatedly sampling random elements from the bin until we get one that forms a self-loop. We can then complete the RAM write by swapping $x$ and $y$ so that they form a coupling pair, as in \cref{alg:ram_write}.

\begin{algorithm}[ht]
  \caption{Writing to a RAM Entry}{\label{alg:ram_write}}
  \DontPrintSemicolon
  \SetKwProg{Fn}{Function}{:}{}
  \SetKwFunction{WriteRAM}{WriteRAM}
  \Fn(\tcp*[f]{Write a new value $v_i$ into word $i$}){\WriteRAM{$i$, $v_i$}} {
    $x \gets A[\Index(i)]$\;
    \If{\ReadRAM{$i$} returns any non-empty value} {
      \tcp*{We need to erase the original value before writing the new value}
      Swap $x$ with $A[\FindPartner{x}]$ \tcp*{Break the coupling pair containing index slot $i$}
    }
    $y \gets$ find an arbitrary self-loop in the $\bk[\big]{(v_i \oplus r_i) + 1}$-th partner bin by random sampling\;
    Swap $x$ with $y$\;
  }
\end{algorithm}

Before continuing, we must verify that RAM writes succeed with high probability in $n$ (i.e., that there exists a self-loop $y$ for the RAM write to use), and that the RAM writes take constant expected time, assuming constant-time access to the retrieval data structure.

\begin{claim}
  \label{claim:ramwrite}
  At any given moment, we have with probability $1 - n^{-\omega(1)}$ that each partner bin contains at least $B/4$ self-loops. Conditioning on this event, each RAM write completes in $O(1)$ expected time.
\end{claim}
\begin{proof}
  Consider the $(k + 1)$-th partner bin for some $k$, and let $v_0, v_1, \ldots, v_{n/4 - 1}$ be the values currently stored in RAM. The number of elements in the $(k + 1)$-th partner bin that are \emph{not} in self-loops is equal to 
  \[|\{i \mid v_i \oplus r_i = k \}|.\]
  Let $X_i = \ind[v_i \oplus r_i = k]$ and $X = \sum_{i=0}^{n/4 - 1} X_i$. Since $r_i$ are chosen independently and uniformly at random, we have that $\Pr[X_i] = 2/m$ and that $X_0, X_1, \ldots, X_{n / 4 - 1}$ are mutually independent. Thus $\E[X] = B/2$, and, by a Chernoff bound,
  \[
    \Pr\Bk*{X > \frac{3B}{4}} \le e^{-\Omega(B)} = n^{-\omega(1)}.
  \]
  Taking a union bound over all $k \in [m/2]$, we have that, with probability $1 - n ^{-\omega(1)}$, every partner bin contains at least $B/4$ self-loops. Finally, since each RAM-write samples random elements from a bin until it finds a self-loop, it follows that the expected time for a RAM-write is $O(1)$.
\end{proof}

Note that the proof of \cref{claim:ramwrite} reveals the role of the $r_i$s in the algorithm. By randomizing the partner bin $v_i \oplus r_i$ that gets used to encode RAM-word $v_i$, we ensure that no partner bin ends up being overused. This, in turn, is what allows us to obtain $O(1)$-time RAM writes.

\paragraph{Storing the retrieval data structure.}
Until now, we have assumed black-box access to the retrieval data structure, i.e., that the data structure is stored externally to our hash table. The final step in the RAM construction is to store the retrieval data structure itself in the RAM that we have implemented. It may seem unintuitive that this would be possible, since the RAM is implemented \emph{using} the retrieval data structure. However, we will show that with some care it is possible to store the retrieval data structure in RAM without introducing any cyclic dependencies in our algorithm.

We store the retrieval data structure of $O(n \log \log n)$ bits in the RAM. Following the implementation of \ReadRAM, we can query the retrieval data structure in $O(1)$ worst-case time. See \cref{alg:retrieval_read}.

\begin{algorithm}[ht]
  \caption{Querys to the Retrieval Data Structure}{\label{alg:retrieval_read}}
  \DontPrintSemicolon
  \SetKwProg{Fn}{Function}{:}{}
  \Fn{\QueryRetrieval{$x$}} {
    Call \ReadRAM to query $x$ on the retrieval data structure
  }
\end{algorithm}

Note that \cref{alg:retrieval_read} is already used in several of the procedures defined earlier, namely \cref{alg:find_partner} and (by extension) \cref{alg:ram_write}. Critically, however, \cref{alg:retrieval_read} is not needed to implement RAM-reads (\cref{alg:ram_read}), so \cref{alg:retrieval_read} does not introduce any cyclic dependencies into our algorithm.

Finally, in order that we can support insertions and deletions of keys later on, we must also show how to update the retrieval data structure in the RAM. We must be careful here, since if we edit the RAM words one by one, the retrieval data structure will briefly be in an inconsistent state, which will cause RAM writes to fail (because they need to access the retrieval data structure). The solution is to first simulate the whole updating process (i.e., determine which keys need to be relocated and where they need to be relocated to) without applying any changes to the physical data structure, and then to apply all changes to the hash table in a batch in the end. See \cref{alg:retrieval_write} for details.

{\LinesNotNumbered
  \begin{algorithm}[ht]
    \caption{Updating the Retrieval Data Structure}{\label{alg:retrieval_write}}
    \DontPrintSemicolon
    \SetKwProg{Fn}{Function}{:}{}
    \SetKwFunction{UpdateRetrieval}{UpdateRetrieval}
    \Fn{\UpdateRetrieval{}} {
      Simulate the update of the retrieval data structure by calling \ReadRAM and \WriteRAM.\\
      In this process:
      \begin{itemize}[nosep, leftmargin=2em]
      \item When \WriteRAM is going to relocate a key, we record the movement without actually\\ moving it.
      \item When \ReadRAM is probing a slot, it sees the slot in the state before the whole update,\\ as we have not applied any changes to the keys and slots.
      \end{itemize}
      After the simulation, apply all key relocations atomically.
    }
  \end{algorithm}
}

At this point, we can confirm that our final algorithm has no cyclic dependencies. In particular, the dependency graph for the procedures is given in Figure \ref{fig:dependency}. What makes the dependencies unintuitive is the interaction between the retrieval data structure and the RAM. Updates to the retrieval data structure require both reads and writes in RAM; writes to RAM require queries to the retrieval data structure; and queries to the retrieval data structure require reads to RAM.

\begin{figure}[ht]
  \centering
  \includegraphics[width = 0.6 \textwidth]{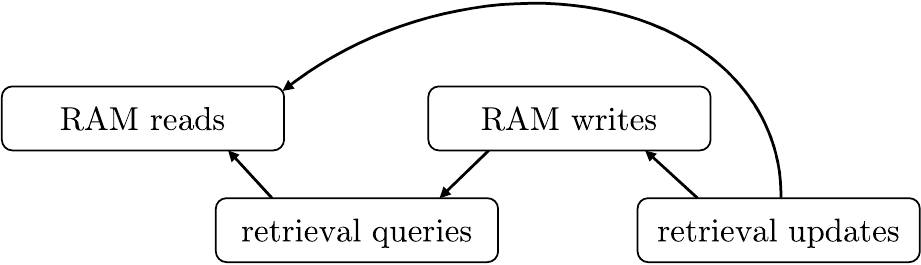}
  \caption{Dependency graph for the procedures.}
  \label{fig:dependency}
\end{figure}

Since each update to the retrieval data structure performs a constant expected number of RAM read/write operations, each requiring constant expected time (assuming no failures), \cref{alg:retrieval_write} also takes constant expected time.

\subsection{Key operations}
\label{sec:key_op}

Having implemented the RAM (and the retrieval data structure), we can now store additional metadata in the RAM in order to efficiently support key operations. 

\paragraph{Key queries.} 
Given a key $x$ that is in the hash table, we can recover $x$'s physical address by computing \FindPartner{$A[\FindPartner{x}]$}. Indeed, regardless of whether $x$ is in a self-loop or a coupling pair, $x$ is guaranteed to be its partner's partner; since $\FindPartner$ returns the physical address of a key's partner, it follows that \FindPartner{$A[\FindPartner{x}]$} is $x$'s physical address. The full procedure for querying a key is given in \cref{alg:key_query}. Note that special handling must be added for the case where $x$ is not in the hash table. 

\begin{algorithm}[ht]
  \caption{Key Queries}{\label{alg:key_query}}
  \DontPrintSemicolon
  \SetKwProg{Fn}{Function}{:}{}
  \SetKwFunction{QueryKey}{QueryKey}
  \Fn{\QueryKey{$x$}} {
    $s_1 \gets$ \FindPartner{$x$}\;
    \If {$A[s_1]$ is empty}{\Return ``$x$ is not in the hash table''}
    $s_2 \gets$ \FindPartner{$A[s_1]$}\;
    \uIf{$A[s_2] = x$} {
      \Return ``$x$ is in slot $s_2$''\;
    } \Else {
      \Return ``$x$ is not in the hash table''\;
    }
  }
\end{algorithm}

\paragraph{Key insertions.}
To support insertions, we will add one additional piece of metadata to RAM: for each bin $k$ we store the number $n_k$ of keys that currently hash to the bin.

Given a key $x$ that is not currently in the hash table, we can insert it by placing it in the $(n_k + 1)$-th slot of bin $k = h(x)$. This maintains the invariant that the keys occupy a prefix of the slots in the bin. Finally, to complete the insertion, we must update the retrieval data structure to store the pair $(x, n_k)$, and we must increment $n_k$. Note that we do not need to consider the case where the bin is already full ($n_k = B$), since we already consider this to be a failure event. The full insertion procedure is given by \cref{alg:key_insertion}.

\begin{algorithm}[ht]
  \caption{Key Insertions}{\label{alg:key_insertion}}
  \DontPrintSemicolon
  \SetKwProg{Fn}{Function}{:}{}
  \SetKwFunction{InsertKey}{InsertKey}
  \Fn{\InsertKey{$x$}} {
    $k \gets h(x)$\;
    $n_k \gets$ the number of keys stored in bin $k$\tcp*{$n_k$ is stored in the RAM}
    Update $A[kB + n_k] \gets x$\;
    Insert key-value pair $(x, n_k)$ into the retrieval data structure\;
    Update $n_k \gets n_k + 1$ in the RAM\;
  }
\end{algorithm}

\paragraph{Key deletions.}
Before we delete a key $x$ that is hashed to bin $k = h(x)$, we can always break the coupling pair containing $x$, so that $x$ becomes a self-loop (we say that we have \emph{decoupled} the pair). Decoupling a pair will erase a word in the RAM, which we make a copy of in a temporary variable. If this word is accessed in the subsequent steps, we access the temporary variable instead. Next, let $y$ be the key stored in the rightmost non-empty slot in bin $k$, i.e., $y = A[kB + n_k - 1]$. If $y$ is in a coupling pair, we decouple that pair as well (and, once again, we store the erased RAM word as a temporary variable). Next, we delete $x$ and move $y$ to the initial place of $x$. (This maintains the invariant that the keys in the bin occupy a prefix of the slots.) Finally, we update the metadata as necessary, and re-encode the RAM words we have broken by decoupling pairs. The procedure is shown in \cref{alg:key_deletion}.

\begin{algorithm}[ht]
  \caption{Key Deletions}{\label{alg:key_deletion}}
  \DontPrintSemicolon
  \SetKwProg{Fn}{Function}{:}{}
  \SetKwFunction{DeleteKey}{DeleteKey}
  \Fn{\DeleteKey{$x$}} {
    \If{$x$ is in a coupling pair} {
      Swap $x$ with its partner. This will erase a word in the RAM. We copy the RAM word's value to a temporary variable. If this erased word is accessed in the following steps, we access the temporary copy instead.\;\label{line:erase_word}
    }
    $k \gets h(x)$\;
    $n_k \gets$ the number of keys stored in bin $k$\;
    \If{$A[kB + n_k - 1]$ is in a coupling pair} {
      Swap $A[kB + n_k - 1]$ with its partner, and make a temporary copy of the erased word\;\label{line:erase_word_2}
    }
    $t \gets \QueryRetrieval{x}$ \tcp*{$x = A[kB + t]$}
    $y \gets A[kB + n_k - 1]$\;
    Remove $x$ from the table, and move $y$ to $A[kB + t]$ (the slot formerly occupied by $x$)\;
    Delete key $x$ from the retrieval data structure\;
    Update the value associated with $y$  in the retrieval data structure to be $t$\;
    \tcp*{The logical address of $y$ is now $kB + t$}
    Update $n_k \gets n_k - 1$ in the RAM\;
    For all words erased in \cref{line:erase_word,line:erase_word_2}, call \WriteRAM to reencode them in RAM\;
  }
\end{algorithm}

\subsection{Putting the pieces together}

In the previous subsections, we have seen the techniques to encode a RAM and support key operations, as long as there is no failure event. Next, we put the pieces together to prove \cref{thm:simple}.

\ThmSimple*

\begin{proof}
  As mentioned in \cref{sec:encode_ram}, we divide all $n$ slots into $m$ bins, each consisting of $B$ slots. We use a hash function $h : U \to [m]$ to map keys to bins. By forming coupling pairs between certain keys, we encode a RAM of $\Theta(n \log n)$ bits with word-size $\Theta(\log n)$. We store some metadata in the RAM, including a retrieval data structure and others, to support key operations, as introduced in \cref{sec:key_op}.

  The procedures in \cref{sec:key_op} perform all key operations correctly in constant expected time, unless any of the following \emph{failure events} occur.
  \begin{itemize}
  \item Bin overflow and underflow: the number $n_k$ of keys that hash to a certain bin $k \in [m]$ becomes larger than $B$ or smaller than $B/2$.
  \item Overusing partner bins: for some $k \in [m/2]$, the number of self-loops in the $(k+1)$-th partner bin goes below $B/4$.
  \end{itemize}
  We detect the failure events by maintaining metadata in the RAM. Specifically, we maintain in the RAM (1)~the number $n_k$ of keys that hash to bin $k$ (these are also used for the key operations in \cref{sec:key_op}), and (2)~the number of self-loops in each partner bin. When any failure event occurs, we reconstruct the entire hash table using a different hash function. Due to \cref{claim:no_overflow_underflow,claim:ramwrite}, both failure events occur with $n^{-\omega(1)}$ probability. Thus, the expected reconstruction time is a negligible $o(1)$.

  Note that our hash table (as presented in this section) supports $n\bk{1 - B^{-1/4}} \pm O(1)$ keys, but, somewhat unintuitively, cannot support fewer than that. (This is one of the issues that we will solve later on with a more sophisticated version of the data structure.) The problem is that the hash table relies on the encoded RAM, which does not work when the load factor becomes too small. Thus, to complete the theorem, we must specify how to initialize the hash table with $n\bk{1 - B^{-1/4}} \pm O(1)$ keys at once. In the initialization process, we compute all metadata that should be stored in the RAM, and encode all RAM words at once (see \cref{alg:init}). The initialization uses $O(n \log n)$ bits of external scratch space (to compute the data that should be encoded in RAM) and takes linear expected time.

  \begin{algorithm}[ht]
    \caption{Initialization}{\label{alg:init}}
    \DontPrintSemicolon
    \SetKwProg{Fn}{Function}{:}{}
    \SetKwFunction{Initialize}{Initialize}
    \Fn(\tcp*[f]{The number of keys is $m = n (1 - B^{-1/4}) \pm O(1)$}){\Initialize{$\{x_0, x_1, \ldots, x_{m-1}\}$}} {
      \For{$i=0$ to $m-1$}{
        Put key $x_i$ into the first empty slot in bin $h(x_i)$
      }
      Compute the sequence of words $(v_0, v_1, \ldots, v_{m/4 - 1})$ that should be encoded in the RAM\;
      \For{$i=0$ to $m/4-1$} {
        Swap $A[\Index(i)]$ with any self-loop in the $\bk[\big]{(v_i \oplus r_i) + 1}$-th partner bin
      }
    }
  \end{algorithm}

  Lastly, in addition to expected constant time, we note that the queries take constant time in the worst case. This is because the query procedure \cref{alg:key_query} and its subroutines \cref{alg:find_partner,alg:retrieval_read} all take $O(1)$ worst-case time.
\end{proof}

\section{A Fixed-Size Hash Table with Strong Guarantees}
\label{sec:fixed_size}

In this section, we show how to eliminate several of the deficiencies of the hash table from the previous section. 

\begin{theorem}
  \label{thm:fixed_size}
      Assuming access to a constant number of $O(1)$-wise independent hash functions, there is an open-addressing hash table that maintains $n$ or $n - 1$ keys in $n$ slots, supports insertions and deletions in $O(1)$ time, with high probability in $n$, and supports queries in $O(1)$ worst-case time. 
\end{theorem}

Comparing to \cref{thm:simple}, this above theorem achieves better guarantees in the following perspectives.
\begin{itemize}
\item We achieve a load factor of $1$ or $1 - 1/n$, instead of $1 - 1 / \poly \log n$ in the previous section.
\item We only need to access $O(1)$-independent hash functions, instead of fully-independent hash functions.
\item All operations take constant time with high probability, instead of only achieving constant expected time.
\end{itemize}

It is worth noting that, even with these improvements, we will will not yet be at our final data structure. In particular, the final data structure (introduced in Section \ref{sec:resizing}) will also support dynamic resizing so that the number of elements can change arbitrarily over time.

The guarantees in the current section are achieved by making three changes to the algorithm introduced in \cref{sec:simple}. Before presenting all of the details, we begin with a high-level summary of the changes below:

\paragraph{Raising the load factor with a backyard.} Recall that in \cref{sec:simple}, every bin has capacity $
B$, while the number of keys hashed to each bin is within $[B - 2B^{3/4},\, B]$ with high probability, which means that a $\Theta(B^{-1/4})$ fraction of slots are left empty. Now, we shrink the capacity of each bin to $B' \defeq B - 2B^{3/4}$. This means that every bin $k \in [m]$ will be completely filled with keys that hash to it, and that there will also be $\Theta(B^{3/4})$ keys that hash to bin $k$ but could not be accommodated due to the capacity -- we call these keys the \emph{overflown keys}. Alongside the $m$ shrinked bins, we add a \emph{backyard} to collect the $\Theta(n/B^{1/4})$ overflown keys from all $m$ bins.

In order to maintain the keys in the backyard, we use $O(\log n)$ bits per key to construct a retrieval data structure that maps every (overflown) key to its address in the backyard. Since the size of the backyard is small enough, this retrieval data structure fits into our encoded RAM of $O(n \log n)$ bits. We also need linked lists to maintain the set of overflown keys that hash to each bin $k$, because when some key in bin $k$ is deleted, we need to ``pull back'' an overflown key from the backyard -- the linked list for bin $k$ will help find such an element.

By introducing the backyard, no slot is empty in the hash table, so we can now support a load factor of up to $1$. In order to simplify the discussion for this section, we will require that the load factor is either $1 - 1/n$ or $1$ at all times (but we will also remove this requirement in the next section).

\paragraph{Increasing the bin size to work with \boldmath$O(1)$-independent hash functions.} In \cref{sec:simple}, we set the bin size $B = \poly \log n$, which, if we assume fully random hash functions, is sufficient to bound the deviation of the number of keys that hash to a certain bin. However, if the hash functions are only $O(1)$-independent, then in order to get usable concentration bounds on the number of keys that hash to each bin, we need a larger bin size. In this section, we will increase the bin size to $B = n^{\Omega(1)}$.

Larger bin sizes bring a new issue: in \cref{sec:simple} we used a retrieval data structure to maintain logical offsets of the keys, which take $O(n \log \log n)$ bits in total, because each offset is $O(\log \log n)$-bit long. However, after increasing the bin size to $B = n^{\Theta(1)}$, each offset within a bin has $\Theta(\log n)$ bits. Thus the space used by the retrieval data structure becomes $\Theta(n \log n)$, which does not clearly fit in a RAM of $\Theta(n \log n)$ bits, because the latter $\Theta(n \log n)$ has a small leading factor.

Fortunately, by setting $B = n^{\delta}$ for a sufficiently small constant $\delta > 0$, we can reduce the leading factor in the space usage of the retrieval data structure, while preserving the leading factor in the size of the RAM. Thus, we can still use a retrieval data structure to maintain the logical offsets.\footnote{The retrieval data structure occupies a constant fraction of our encoded RAM. It is worth noting that, with a bit more involvement, there are also more space efficient solutions that one could use instead of the retrieval data structure. In particular, one could organize the keys in each bin so that their logical addresses have the same layout as if the bin was implemented as the high-performance hash table in \cite{bender2022optimal}; this would allow for the keys' logical addresses to be recovered using only $O(n \log \log n)$ bits of additional metadata stored in RAM.}

\paragraph{Dense and sparse RAMs.} In \cref{sec:simple}, we have achieved constant expected time per operation. Actually, there is only one step that is not with high-probability constant time: When we write a value $v_i$ to the RAM word $i$, we need to find a self-loop in the $((v_i \oplus r_i) + 1)$-th partner bin, which will form a coupling pair with the key in $A[\Index(i)]$. Such a self-loop is found by sampling random elements in the bin and checking if they are self-loops. Since there are a constant fraction of self-loops in every partner bin, the sampling process will succeed in $O(1)$ rounds in expectation -- but not with high probability.

To make the RAM writes high-probability constant-time, we divide all bins into two identical groups, each encoding a \emph{basic RAM} of $O(n)$ words using the construction introduced in \cref{sec:simple}. The first group implements the \emph{dense RAM} that stores most of the information, but writing a word to the dense RAM is not high-probability constant-time. The second group implements a \emph{sparse RAM}, in which most of the words are empty (i.e., most of the index slots contain self-loops), while only $O\bk[\big]{\sqrt{B}} = O(n^{\delta/2})$ words have values. The sparse RAM serves as a buffer for the recent updates to the RAM words, which have not been applied to the dense RAM yet. This idea of buffering work on a small but fast secondary data structure has also been used (albeit in different forms) in many past works on worst-case constant-time hash tables \cite{dietzfelbinger1990new, ArbitmanNaSe10, arbitman2009amortized,bender2024modern}.

Since the sparse RAM only has $O\bk[\big]{\sqrt{B}}$ coupling pairs, in any given partner bin in the sparse RAM, at least a $(B - O(\sqrt{B}))/B = 1 - O(1/n^{\delta/2})$ fraction of elements will be self-loops. Therefore, to find a self-loop in a given bin, random sampling takes $O(1)$ time with high probability, which implies that we can write a word to the sparse RAM in $O(1)$ time with high probability.

In order to make an update to the RAM, we first write the update to the sparse RAM. Then, we spend a constant amount of work to apply the buffered updates stored in the sparse RAM to the dense RAM. With high probability, at any given moment, the sparse RAM will only contain a small number of buffered updates.

\bigskip

With these three modifications, we can achieve the guarantees in \cref{thm:fixed_size}. Throughout the main proof, we will assume the universe $U$ has size $\poly n$, which we know from \cref{sec:prelim} is without loss of generality.

\subsection{The full construction}

Let $B = \Theta(n^{\delta})$ where $\delta \in (0, 0.1)$ is a constant to be determined later. We divide all $n$ slots into
\begin{itemize}
\item a \emph{backyard} of $\Theta(n / \log^{10} n)$ slots; and
\item $m$ bins of $B$ slots each, where $mB = n - \Theta(n / \log^{10} n)$, and where $m$ is a power of two. All these bins are called the \emph{frontyard}.
\end{itemize}
We label the bins with $0, \ldots, m-1$, and number the slots in bin $k$ with integers in $[kB, (k+1)B)$. The slots in the backyard are labeled with integers in $[mB, n)$, so the numbers for all slots cover $0, \ldots, n-1$.

The storage structure is similar to \cref{sec:simple}. There is a \emph{$c$-wise independent hash function} $h : U \to [m]$ that maps every key into one of the slots, where $c = \Theta(1)$ is sufficiently large relative to $\delta^{-1}$. In the logical layout, bin $k$ should be filled with keys that are hashed to bin $k$. We denote by $n_k$ the number of keys hashed to bin $k$. We will show that with high probability, $n_k \ge B$ at any moment. The $n_k - B$ of these keys that cannot be accommodated in bin $k$ in the logical layout will be stored in the backyard in arbitrary order. The hash table fails when $n_k < B$.

\begin{claim}
  \label{claim:no_underflow}
 We have with high probability in $n$ that the number of keys hashed to every bin is at least $B$.
\end{claim}

\begin{proof}
  We first bound the probability that bin 0 ``underflows'' (less than $B$ keys hash to it). Let $X_i = \ind\Bk{\text{the $i$-th key $\hashto$ bin 0}}$ and $X \defeq \sum_{i=1}^{n} X_i = n_k$. We know $\E[X] = B \cdot \bk[\big]{1 + \Theta(1 / \log^{10} n)}$ and $\Var(X) \le \E[X]$. By concentration bounds for $c$-wise independent variables (see, e.g., Theorem 4 in \cite{schmidt1995chernoffhoeffding}), we have
  \begin{align*}
    \Pr\Bk*{\abs[\big]{X - \E[X]} \ge \frac{B}{\log^{10} n}}
    &\le O\bk*{\frac{\E[X]}{\bk[\big]{B/\log^{10} n}^{2}}}^{c/2}
    \le \bk*{\frac{n^{\delta(1 + o(1))}}{n^{2\delta(1 + o(1))}}}^{c/2}
    < n^{-\delta(c/2)(1 \pm o(1))} \ll n^{-\delta c/3}.
  \end{align*}
  So, bin 0 underflows with probability $n^{-\delta c/3}$. Taking a union bound over all $m$ bins, we know that with probability $1 - n^{-\delta c/3 + 1}$, no bin will underflow. Since $c$ is a sufficiently large constant (with respect to  $\delta^{-1}$), this is a high-probability bound in $n$.
\end{proof}

For the moment, we assume access to a retrieval data structure (see \cref{thm:retrieval}) that maps every key $x$ stored in the frontyard to its logical offset within the bin. Since every logical offset takes $\delta \log n$ bits to encode, the retrieval data structure uses $\delta \cdot O(n \log n)$ bits of space. Later, we will store the retrieval data structure in the RAM we implement, which has $\ge 0.08 n \log n$ bits. By choosing a small enough $\delta$, the retrieval data structure fits in the RAM.

The bins are divided into 2 groups: the first $m/2$ bins form the first group (group 0), and the last $m/2$ bins form the second group (group 1). Each group will independently implement a \emph{basic RAM}. To encode a basic RAM, we swap certain pairs of keys to form \emph{coupling pairs} as in \cref{sec:simple}. After implementing both basic RAMs, we will let them play the roles of \emph{dense RAM} and \emph{sparse RAM} respectively, which together provide the functionality of an \emph{advanced RAM} (i.e., a RAM with high-probability constant-time operations).

\paragraph{Basic RAM within each group.}
We first describe how to implement the basic RAM used in each group. The approach is the same as \cref{sec:simple}, except that now we use larger bins and $O(1)$-independent hash functions. By symmetry, we focus on the first group, i.e., bins $0$ to $m/2 - 1$. 

Each bin consists of $B$ slots, which contains $B$ keys hashing to it in the logical layout. The first $m/4$ bins are regarded as \emph{index bins}, while the last $m/4$ bins are regarded as \emph{partner bins}. We say the first half of the slots in each index bin are \emph{index slots}, i.e., the slots in index bins with offset $t < B/2$. There are $mB/8$ index slots. We let $\Index(i)$ denote the slot number of the $(i + 1)$-th index slot.

Let $r_0, \ldots, r_{mB/8 - 1}$ be $c$-wise independent random numbers in $[m/4]$. They are obtained via a $c$-wise independent hash function $r : [mB/8] \to [m/4]$.

The basic RAM will encode a sequence of words $v_0, \ldots, v_{mB/8-1} \in [m/8] \cup \BK{\bot}$, where ``$\bot$'' denotes that the word is empty. For each word $i$, if $v_i$ is not empty, we find key $x$ that is logically stored in $A[\Index(i)]$, and swap it with an arbitrary self-loop $y$ in the $((v_i \oplus r_i) + 1)$-th partner bin. If $v_i = \bot$, we just leave $x = A[\Index(i)]$ as a self-loop.

\cref{alg:basic_ram} shows the procedures to read or write a word to the basic RAM. The procedure of writing calls \FindPartner (see~\cref{alg:find_partner}) as a subroutine, which further queries the retrieval data structure.

\begin{algorithm}[ht]
  \caption{Basic RAM Operations}{\label{alg:basic_ram}}
  \DontPrintSemicolon
  \SetKwProg{Fn}{Function}{:}{}
  \SetKwFunction{ReadBasicRAM}{ReadBasicRAM}
  \Fn(\tcp*[f]{$g \in \{0, 1\}$ is the group number}){\ReadBasicRAM{$g, i$}} {
    $x \gets A[gmB/2 + \Index(i)]$ \tcp*{the $i$-th index slot in group $g$}
    \uIf(\tcp*[f]{$x$ is hashed to the $(k + 1)$-th partner bin}){$k \defeq h(x) - m/4 - gm/2 \ge 0$} {
      \Return $k \oplus r_i$ \tcp*{$r_i$ are $c$-wise independent random integers in $[m/4]$}\label{line:rand_shift/basic}
    }\Else(\tcp*[f]{$x$ is a self-loop}){
      \Return $v_i = \bot$ \tcp*{word $i$ is empty}
    }
  }
  \SetKwFunction{WriteBasicRAM}{WriteBasicRAM}
  \Fn(\tcp*[f]{Write a new value $v_i$ into word $i$}){\WriteBasicRAM{$g$, $i$, $v_i$}} {
    $x \gets A[gmB/2 + \Index(i)]$ \tcp*{the $i$-th index slot in group $g$}
    \If(\tcp*[f]{need to erase the original value}){\ReadBasicRAM{$g, i$} returns any non-empty value} {
      Swap $x$ with $A[\FindPartner{x}]$ \tcp*{Break the coupling pair containing index slot $i$}
    }
    \If{$v_i \ne \bot$} {
      $y \gets$ find a self-loop in the $\bk[\big]{(v_i \oplus r_i) + 1}$-th partner bin (in group $g$) by random sampling \label{line:sampling/basic} \;
      Swap $x$ with $y$ \;
    }
  }
\end{algorithm}

This construction works as long as every partner bin contains a constant fraction of self-loops -- this is for \cref{line:sampling/basic} to find a self-loop in constant expected time. We show that this condition holds with high probability in $n$.

\begin{claim}
  \label{claim:ramwrite/basic}
  At any given moment, with high probability in $n$, each partner bin contains at least $B/4$ self-loops. Conditioning on this event, each RAM write completes in $O(1)$ expected time.
\end{claim}

\begin{proof}
  We consider the $(k+1)$-th partner bin for a certain $k \in [m/4]$, and recall that $v_0, \ldots, v_{mB/4 - 1}$ are the sequence of words encoded by the basic RAM. The number of non-self-loop elements in the bin equals
  \[
    |\BK{i \mid v_i \oplus r_i = k}|.
  \]
  We define $X_i = \ind[v_i \oplus r_i = k]$ and $X = \sum_{i=0}^{mB/8 - 1} X_i$. Since $r_0, \ldots, r_{mB/8 - 1}$ are $c$-wise independent, we know $\Pr[X_i] = 4/m$ and $X_0, \ldots, X_{mB/8 - 1}$ are also $c$-wise independent, thus $\E[X] = B/2$. By standard concentration bounds (e.g., Theorem 4 in~\cite{schmidt1995chernoffhoeffding}), we have
  \[
    \Pr\Bk*{X > \frac{3B}{4}} \le O\bk*{\frac{1}{B}}^{c/2} = O(n^{-\delta c/2}).
  \]
  Taking a union bound over $m/4$ partner bins, we know that with probability $1 - \Omega(n^{-\delta c / 2 + 1})$, every partner bin has at least $B/4$ self-loops. Since $c$ is a sufficiently large constant multiple of $\delta^{-1}$, this is a high-probability bound in $n$.
\end{proof}

The basic RAM reads take constant time in the worst case. Writes, on the other hand, take $O(1)$ time in expectation, but not with high probability. However, when the basic RAM is sparse enough, writes also take constant time with high probability:

\begin{lemma}
  \label{lem:sparse_ram}
  In a basic RAM, if among all encoded words $v_0, \ldots, v_{mB/4 - 1}$, at most $O\bk[\big]{\sqrt B} = O(n^{\delta / 2})$ words are non-empty, then \WriteBasicRAM takes $O(1)$ time to finish with high probability in $n$.
\end{lemma}

\begin{proof}
  The time bottleneck of \WriteBasicRAM is to find a self-loop in a certain bin. Because there are only $O\bk[\big]{\sqrt B}$ coupling pairs in the whole basic RAM, each sample will fail with probability $\le O\bk[\big]{1/\sqrt B}$. The probability that the sampling does not terminate within $c \gg 1/\delta$ rounds is $O(B^{-c/2}) = O(n^{-\delta c/2})$, i.e., \WriteBasicRAM takes $O(c) = O(1)$ time with high probability in $n$.
\end{proof}

\paragraph{Advanced RAM.} 
As we mentioned earlier, the first basic RAM serves as the \emph{dense RAM} and the second one serves as the \emph{sparse RAM}. They will together maintain a sequence of words $v_0, v_1, \ldots, v_{mB/10 - 1} \in [m/4]$, which we call the \emph{advanced RAM}, supporting reads in $O(1)$ worst-case time and writes in $O(1)$ time with high probability in $n$.

The first basic RAM -- the \emph{dense RAM} -- stores a sequence of words $w_0, w_1, \ldots, w_{mB/10 - 1} \in [m/4]$. They represent an \emph{outdated} state of the advanced RAM, where a small number of words might be different from the advanced RAM.

The second basic RAM -- the \emph{sparse RAM} -- stores $mB/8$ words from $[m/4] \cup \{\bot\}$. The first $mB/10$ words, namely $\tilde w_0, \ldots, \tilde w_{mB/10 - 1}$, record the updates that we should apply to the dense RAM. We call them \emph{buffered updates}. If $\tilde w_i \ne \bot$ for some word $i \in [mB/10]$, the corresponding word $w_i$ in the dense RAM is outdated -- word $v_i$ in the advanced RAM should equal to the buffered update $\tilde w_i$ instead. Otherwise, when $\tilde w_i = \bot$, word $v_i$ in the advanced RAM equals the dense RAM word $w_i$. The last $mB / 40$ words in the sparse RAM, which are \emph{not} buffered updates, store a queue of all locations $i \in [mB/10]$ with buffered updates $\tilde w_i \ne \bot$. It is called a \emph{buffer queue}. See \cref{fig:advanced_RAM}.

\begin{figure}[ht]
  \centering
  \includegraphics[width = 0.7 \textwidth]{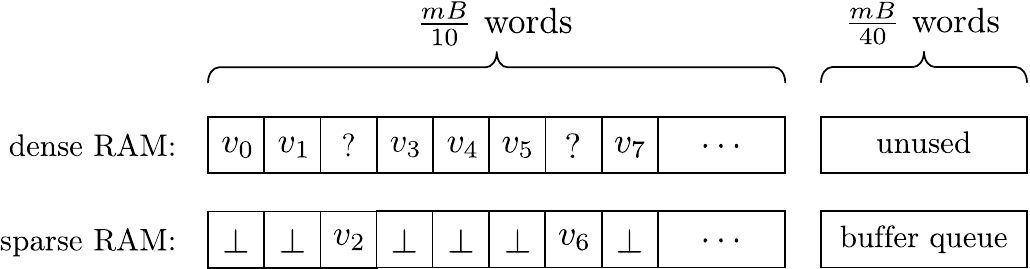}
  \caption[Advanced RAM]{Storing the advanced RAM using two basic RAMs. Each word $v_i$ ($i \in [mB/10]$) in the advanced RAM is either (1) stored in word $i$ of the dense RAM, while word $i$ of the sparse RAM remains empty; or (2) stored in word $i$ of the sparse RAM, while word $i$ of the dense RAM can be arbitrary. The last $mB/40$ words of the sparse RAM maintains the \emph{buffer queue}, which contains all locations $i \in [mB/10]$ with buffered updates in the sparse RAM.}
  \label{fig:advanced_RAM}
\end{figure}

\smallskip
When we read a word $v_i$ in the advanced RAM, we can read both words $w_i$ and $\tilde w_i$ in the dense RAM and the sparse RAM. This will determine the value of $v_i$ in worst-case $O(1)$ time.

When we update a word $v_i$, we first update the sparse RAM word $\tilde w_i$ using the new value, and also push $i$ into the buffer queue. This step takes $O(1)$ time with high probability in $n$. Then, as long as there are buffered updates in the sparse RAM, we spend $O(1)$ time to apply these updates to the dense RAM, thus making these words in the sparse RAM empty again. See \cref{alg:advanced_ram} for details.

\begin{algorithm}[ht]
  \caption{Advanced RAM Operations}{\label{alg:advanced_ram}}
  \DontPrintSemicolon
  \SetKwProg{Fn}{Function}{:}{}
  \SetKwFunction{ReadAdvancedRAM}{ReadAdvancedRAM}
  \Fn{\ReadAdvancedRAM{$i$}} {
    $\tilde w_i \gets$ \ReadBasicRAM{$1, i$} \tcp*{Read the sparse RAM}
    \If{$\tilde w_i \ne \bot$} {
      \Return $\tilde w_i$ \;
    }
    \Return \ReadBasicRAM{$0, i$} \tcp*{Read the dense RAM}
  }

  \SetKwFunction{WriteAdvancedRAM}{WriteAdvancedRAM}
  \Fn(\tcp*[f]{Write a new value $v_i$ into word $i$}){\WriteAdvancedRAM{$i$, $v_i$}} {
    \If{\ReadBasicRAM{$1, i$} $= \bot$} {
      Push $i$ into the buffer queue \tcp*{New buffered update}
    }
    \WriteBasicRAM{$1, i, v_i$} \tcp*{First write the update into the sparse RAM}
    \While{buffer queue is non-empty} { \label{line:while}
      $j \gets$ buffer queue's front element \;
      $v_j \gets$ \ReadBasicRAM{$1, j$} \;
      Try \WriteBasicRAM{$0, j, v_j$}; give up when we have taken 100 random samples in total during the current \WriteAdvancedRAM execution \label{line:try_move} \;
      \uIf{\WriteBasicRAM has given up} {
        \Return \tcp*{No change to the sparse RAM}
      } \Else {
        \WriteBasicRAM{$1, j, \bot$} \tcp*{Delete the buffered update}
        Pop $j$ from the buffer queue \;
      }
    }
  }
\end{algorithm}

\smallskip

On \cref{line:try_move}, we try to perform \WriteBasicRAM to the dense RAM, moving buffered updates to the dense RAM. The most time-consuming step in \WriteBasicRAM is to find any self-loop in a specific bin -- it repeatedly samples a random element in the bin, and terminates when the sampled element is a self-loop, which occurs with at least $1/4$ probability in each round. When we run the random sampling process, we count the number of elements we have sampled so far in the entire while-loop (on \cref{line:while}), and stop when it reaches 100. Thus, the while-loop takes $O(1)$ time in the worst case, and the expected number of buffered updates that it moves is $\ge 25$.

The advanced RAM fails when the sparse RAM contains more than $\Omega\bk[\big]{\sqrt B} = \Omega(n^{\delta/2})$ buffered updates, in which case the premise of \cref{lem:sparse_ram} no longer holds. We show that this occurs with low probability.

\begin{claim}
  \label{claim:adv/no_fail}
  At any given moment, the sparse RAM contains at most $O(n^{\delta/2})$ buffered updates with high probability in $n$.
\end{claim}

\begin{proof}
  Suppose there are $k$ buffered updates after the $t$-th \WriteAdvancedRAM operation, and let $0 \le j \le t$ be the smallest $j$ such that, after the $(t - j)$-th operation, there were no buffered updates. We let $X_i$ denote the number of buffered updates processed during the $t$-th operation, so we know $0 \le X_i \le 100$ and $\E[X_i] \ge 25$. In the time interval $I = (t - j, \, t]$, the number of operations performed is $|I| = j$, and the number of buffered updates that have successfully been processed is $\sum_{i \in I} X_i$, so
  \begin{equation*}
  j = k + \sum_{i \in (t - j, t]} X_i.
  \end{equation*}
  To bound the probability of $k \ge \Omega(n^{\delta / 2})$, it suffices to bound the probability that, for some $j \ge 0$,
  \begin{equation}n^{\delta / 2} + \sum_{i \in (t - j, \, t]} X_i \le j.\label{eq:jk}
  \end{equation}
  This inequality is only possible for $j \ge n^{\delta/2}$. However, for $j \ge n^{\delta/2}$, by a Chernoff bound,
  \[\Pr\left[\sum_{i \in (t - j, \, t]} X_i \le j\right] = e^{-\Omega(j)}.\]
  Summing over $j \ge n^{\delta/2}$, we can bound the probability of \eqref{eq:jk} by $\sum_{j \ge n^{\delta/2}} e^{-\Omega(j)} \ll 1 / \poly(n)$.
\end{proof}

As long as the sparse RAM is sparse, \WriteBasicRAM on the sparse RAM takes constant time with high probability in $n$; moreover, the buffer queue always fits in the last $mB/40$ words in the sparse RAM. Therefore, the total running time for \WriteAdvancedRAM is also a constant with high probability, conditioned on the advanced RAM not failing.

\paragraph{Storing the retrieval data structure.}

So far, we have assumed black-box access to the retrieval data structure that stores the logical offsets of all keys in the frontyard (i.e., in the bins). Similar to \cref{sec:simple}, the next step is to store the retrieval data structure in the advanced RAM that we have constructed. \QueryRetrieval and \UpdateRetrieval will be implemented in the same way as \cref{alg:retrieval_read,alg:retrieval_write}, with \ReadAdvancedRAM and \WriteAdvancedRAM replacing \ReadRAM and \WriteRAM, respectively. \QueryRetrieval only relies on \ReadAdvancedRAM, while \UpdateRetrieval relies on both reading and writing the advanced RAM. As shown in \cref{fig:dependency/adv}, we do not incur cyclic dependencies.

\begin{figure}[ht]
  \centering
  \includegraphics[width = 0.5 \textwidth]{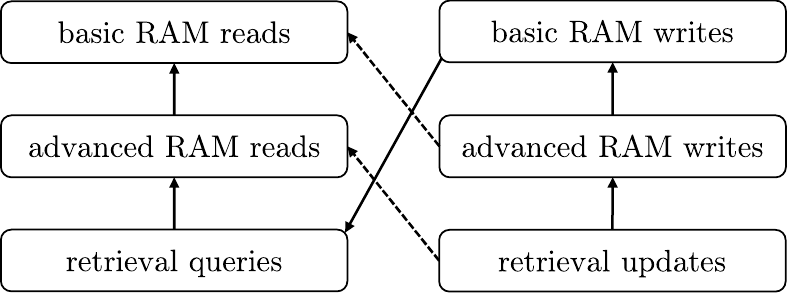}
  \caption{Dependency graph for the advanced RAM. Solid lines show a topological order.}
  \label{fig:dependency/adv}
\end{figure}

\smallskip

We conclude the construction of the advanced RAM by comparing the sizes of the retrieval data structure and the advanced RAM. The advanced RAM consists of $mB/10$ words, each storing a value in $[m/4]$, so it has at least
\[
  \frac{mB}{10} \cdot \log \frac{m}{4} \ge (1 - o(1)) \frac{n}{10} \log n^{1-\delta} \ge (1 - o(1)) \frac{n}{10} \log n^{0.9} > 0.08 n \log n
\]
bits. In contrast, according to \cref{thm:retrieval}, the number of bits in the retrieval data structure does not exceed $\beta (n \log B + n \log \log \frac{U}{n}) = (1 + o(1)) \cdot \beta \delta n \log n$. By choosing $\delta$ to be a small constant, say $1/(1000\beta)$, we know that the retrieval data structure occupies at most $0.01 n \log n$ bits of memory, which fits in the advanced RAM of $0.08 n \log n$ bits.

In summary, we have constructed an advanced RAM, together with a retrieval data structure that stores the logical offset for all keys in the frontyard. Besides storing the retrieval data structure, the advanced RAM provides $\ge 0.07 n \log n$ bits of additional space, which can support other metadata used to manage the keys (we will describe this metadata in a moment). Both reading and writing to the advanced RAM takes $O(1)$ time with high probability in $n$, and moreover, each read operation takes $O(1)$ time in the worst case.

\paragraph{Metadata for the backyard. }
In order to implement the backyard, we need to store some additional metadata in the advanced RAM. 

For each bin $k$, we maintain an \emph{overflow linked-list} of the overflow keys for that bin. The overflow linked lists are stored together (in advanced RAM) as an array of size $n/B + T$, where $T = O(n / \log^{10} n)$ is the size of the backyard, and where each entry of the array is $O(\log n)$ bits. The last $n / B$ entries of the array store the base pointers for the $n/B$ overflow linked lists, and the earlier entries store the internal nodes. If an internal node represents an element that is in some position $j$ of the backyard, then the internal node is, itself, stored in position $j$ of the array. The overflow linked lists take $O(T \log n)$ bits of space, where $T$ is the size of the backyard. 

Additionally, so that we can perform queries in the backyard, we store a \emph{backyard retreival data structure} mapping each backyard key to its position in the backyard. By \cref{thm:retrieval}, this retrieval data structure takes $O(T \log |U|) = O(T \log n)$ bits of space, where, again $T$ is size of the backyard. Since $T = O(n / \log^{10} n)$, the total size of these additional data structures is $O(n / \log^{9} n)$ bits.

Finally, if there is a free slot in the backyard (because there are $n - 1$ total keys present), we also store the position of that free slot in advanced RAM.

\paragraph{Key operations.} We can support key operations (insertions, deletions and queries) exactly as in \cref{sec:key_op}, but with the following modifications in order to handle the role of the backyard. 

Queries query both the frontyard, using the protocol from \cref{sec:key_op}, and the backyard, using the backyard retrieval data structure. This takes $O(1)$ worst-case time.

Whenever a key is inserted, it is placed directly in the backyard (which will, prior to the insertion, contain exactly one free slot). Finally, to complete the insertion, the appropriate overflow linked list and the backyard retrieval data structure are updated in $O(1)$ time.

Deletions, in turn, are implemented as follows. If the key is in the backyard, then we simply remove the key, update the appropriate overflow linked list, and update the backyard retrieval data structure appropriately, all in $O(1)$ time. If the key is in a frontyard bin $k$, then we delete the key using the protocol from \cref{sec:key_op}, we use the overflow linked list for bin $k$ to find a key $y$ satisfying $h(y) = k$ in the backyard, we delete $y$ from the backyard (using the deletion protocol for the backyard described above), and we insert $y$ into the frontyard (using the insertion protocol from \cref{sec:key_op}). If there are no overflow keys in the backyard, then we have experienced a low-probability failure (\cref{claim:no_underflow}), and we rebuild the entire data structure from scratch using new hash functions.

All operations take $O(1)$ worst-case time, with high probability. The only case where an operation takes $\omega(1)$ time is we experience a $1 / \poly(n)$-probability failure event (a bin overflows or a basic RAM fails), in which case we spend $O(n)$ time rebuilding the data structure from scratch using new hash functions. Putting the pieces together, we have Theorem \ref{thm:fixed_size}.

\section{Dynamic Resizing} \label{sec:resizing}

Finally, in this section, we show how to extend our data structure to support dynamic resizing, with load factor one at all times:

\begin{theorem}
  \label{thm:resizing}
  There exists a positive constant $c$ such that, given $c$ hash functions that are $c$-wise independent, one can construct an open-addressing hash table with the following guarantees. If, at any given moment, there are $n$ keys in the hash table, then the hash table uses $n$ slots of space; the hash table supports insertions and deletions in time $O(1)$, with high probability in $n$; and the hash table supports queries in time $O(1)$ in the worst case.
\end{theorem}

Note that, throughout this section, $n$ will be a dynamically-changing quantity, as it represents the number of keys present at any given moment. In contrast, when referring to hash table with fixed capacity, we will use $N$ to denote its capacity.

As a convention, we will refer to the data structure from Section \ref{sec:fixed_size} as \emph{fixed-size partner hashing}. Recall that, if fixed-size partner hashing is parameterized to have capacity $N$, then as part of the construction, it also implements a $\Theta(N \log N)$-bit RAM. As in Section \ref{sec:fixed_size}, it will be important that the leading constant in the $\Theta(N \log N)$ is independent of the parameter $\delta \in (0, 0.1)$. As noted in \cref{sec:prelim}, we will sometimes abuse asymptotic notation by using $\delta \cdot O(f(n))$ to denote $\delta \cdot g(n)$ for some $g(n) = O(f(n))$ where the hidden constant is oblivious to $\delta$.

In our applications of fixed-size partner hashing in this section, we will often find ourselves storing different information in the RAM than we did in the previous section. As a slight abuse of notation, we will still refer to such a data structure as an instance of fixed-size partner hashing (just storing different content in its RAM). An important special case will be the one where the RAM is empty (i.e., the logical layout and the physical layout of the data structure are the same).

Finally, before we can describe our approach to resizing, it will be helpful to choose a specific array layout for fixed-size partner hashing. Consider an instance of fixed-size partner hashing with capacity $N$, and that is implemented using a backyard of some size $T = O(N / \log^{10} N)$ and using $M = \Theta((N - T) / B)$ bins of size $B = \Theta(N^{\delta})$. Throughout this section, we shall assume that the data structure is laid out as follows: the first $T$ slots of the array are used to store the backyard, and the remaining $N - T$ slots are used to implement the bins, with each slot $i$ allocated to bin $(i \bmod {M})$. 

The advantage of this array layout is that, if we consider any \emph{prefix} of the array, we can think of the prefix as also consisting of a backyard of size (up to) $T$, and then $M$ bins that are equal-size (up to $\pm 1$). If every element in the frontyard is physically in the bin to which it logically belongs, then we say that the prefix is an \emph{valid prefix of an empty $n$-size RAM}.

\paragraph{High-level approach to resizing.}
Let $n$ denote the \emph{current} number of elements. At any given moment, we will maintain the structural invariant that, with high probability in $n$, the array is broken into segments as follows:
\begin{itemize}
\item The first segment $X$ is a prefix of the array with a power-of-two size satisfying $|X| \in \delta n \cdot [0.1, 0.9]$. There are no restrictions on the order of elements within $X$.
\item For all $i \in [\,\log |X|,\, \lfloor \log n \rfloor)$, the subarray consisting of slots $Y_i := [2^i, 2^{i + 1})$ is a valid $2^i$-size RAM. Furthermore, at any given moment, there is some RAM $Y_{\overline{i}}$, with $\overline{i} \in \{\lfloor \log n \rfloor - 5, \lfloor \log n \rfloor - 4, \lfloor \log n \rfloor - 3\}$, that is considered the \emph{core RAM}, and that is responsible for storing the metadata used by the entire resizable hash table. 
\item The final $n - 2^{\lfloor \log n \rfloor}$ slots form a valid prefix of an empty $2^{\lfloor \log n \rfloor}$-size RAM. We refer to this segment of the array as $Z$.
\end{itemize}

To complete the picture, let us also describe the metadata associated with each of $X, \{Y_i\}, Z$ (as we noted above, all of this metadata will be stored in the core RAM $Y_{\overline{i}}$). For the $O(1) \cdot \delta n$ elements in $X$, we store a retrieval data structure that maps each such element to its position. For each $Y_i$ and for $Z$, we store all of the metadata that would normally be associated with a fixed-size partner hash table. The retrieval data structure for $X$ consumes $\delta \cdot O(n \log n)$ bits, and the metadata for the $Y_i$ and for $Z$ consume a total of $\delta \cdot O(n \log n)$ bits as well (by the analysis in \cref{sec:fixed_size}). 

In addition to this metadata, we store a few other basic pieces of information, namely, the values of $\overline{i}$ and of $n$. We will also, later on in the section, store additional metadata (increasing the overall amount by an $O(1)$-factor independent of $\delta$) in order to support the gradual rebuild work that is necessary in order to maintain the structural invariant described above.

\paragraph{How to think about $Z$.}
Before we continue, it is worth commenting on how to think about $Z$. At any given moment, $Z$ will be a \emph{prefix} of a valid empty $2^{\lfloor \log n \rfloor}$-sized RAM. As insertions and deletions occur, $Z$'s size will change. When $Z$ is large enough that it contains slots allocated to bins (rather than just backyard slots), this means that the sizes of $Z$'s bins also change over time. When the size of one of $Z$'s bins increases (because a slot is added to the end of $Z$), we will pull an element from the backyard into that bin to occupy that slot. Likewise, when the size of one of $Z$'s bins decreases (because a slot is removed from the end of $Z$), we will evict an element from that bin into the backyard. In both cases, of course, will also update the appropriate metadata for $Z$ (i.e., the overflow linked list for the affected bin, and the retrieval data structures used to locate elements in both the backyard and frontyard).

It is worth verifying that, at any given moment, (with high probability in $n$) for each bin $j$ in $Z$, we will be able to fill bin $j$ with elements that hash to bin $j$. The proof is just a minor extension of \cref{claim:no_underflow}.

\begin{claim}
  \label{claim:resizeunderflow}
Let $T = n / \log^{10} n$ be the capacity of $Z$'s backyard, and suppose $|Z| \ge T$. Let $B$ be the current bin size that all of the bins in $Z$ have (up to $\pm 1$). We have with high probability in $n$ that each bin $j$ in $Z$ has at least $B + 1$ keys that hash to it.
\end{claim}

\begin{proof}
  Throughout the proof, we consider only the $n' \defeq |Z| = n - 2^{\lfloor \log n \rfloor}$ keys in $Z$. We first bound the probability that bin 0 ``underflows'' ($B$ or fewer keys hash to it). Let $J_i = \ind\Bk{\text{the $i$-th key $\hashto$ bin 0}}$ and $J \defeq \sum_{i=1}^{n'} J_i$. Letting $T$ denote the size of the backyard and $M$ denote the number of bins, we know $\E[J] = B  + T / M = B + \Theta(n^{\delta} / \log^{10} n) \ge B + n^{\delta - o(1)}$ and $\Var(J) \le \E[J]$. By concentration bounds for $c$-wise independent variables (see, e.g., Theorem 4 in \cite{schmidt1995chernoffhoeffding}), we have
  \begin{align*}
    \Pr\Bk*{\abs[\big]{J - \E[J]} \ge n^{\delta - o(1)}}
    &\le O\bk*{\frac{\E[J]}{\bk[\big]{n^{\delta - o(1)}}^{2}}}^{c/2}
    \le \bk*{\frac{n^{\delta(1 + o(1))}}{n^{2\delta(1 - o(1))}}}^{c/2}
    < n^{-\delta(c/2)(1 \pm o(1))} \ll n^{-\delta c/3}.
  \end{align*}
  So, bin 0 underflows with probability $n^{-\delta c/3}$. Taking a union bound over all $M \le n$ bins, we know that with probability $1 - n^{-\delta c/3 + 1}$, no bin will underflow. Since $c$ is a sufficiently large constant (with respect to  $\delta^{-1}$), this is a high-probability bound in $n$.
\end{proof}

\begin{remark}[High-Probability Guarantees When $|Z|$ is Small.]
A small amount of special care is needed for the case where $|Z| = m^{o(1)}$. In this case, the retrieval data structure used to access the elements in $Z$ would have size $m^{o(1)}$, and would naively fail to offer high-probability-in-$m$ guarantees (since, a priori, the retrieval data structure offers high-probability guarantees as a function of its own size). However, by adding $\sqrt{m}$ dummy elements to the retrieval data structure, we can avoid this issue and always maintain high-probability-in-$m$ guarantees.
\end{remark}

\paragraph{Gradually rebuilding in order to maintain the high-level structure.}
At a high level, as insertions and deletions are performed, the size of $Z$ will change. Whenever $\lfloor \log n \rfloor$ decreases, the rightmost $Y_i$ becomes the new $Z$. Likewise, whenever $\lfloor \log n \rfloor$ increases, $Z$ becomes a $Y_i$, and a new (initially empty) $Z$ is created. 

During each insertion/deletion, we will spend $O(1)$ time performing what we call \emph{migration work}, which will allow us to maintain the invariants that (1) $X$ has a power-of-two size in the range $\delta n \cdot (0.1, 0.9)$, (2) $Z$ is a valid prefix of an empty $2^{\lfloor \log n \rfloor}$-size RAM, and (3) there is some core-RAM $Y_{\overline{i}}$, satisfying $\overline{i} \in \{\lfloor \log n \rfloor - 5, \lfloor \log n \rfloor - 4, \lfloor \log n \rfloor - 3\}$, storing all of the metadata for the data structure. The migration work is spent gradually changing the state of the data structures that the invariants are true at all times. 

In more detail, the migration work is spent on a few different tasks. 
\begin{itemize}
    \item During each insertion/deletion, we spend $O(1)$ migration work copying metadata from $Y_{\overline{i}}$ to each of $Y_{\overline{i} + 1}$ and $Y_{\overline{i} - 1}$. This allows us to re-decide the value of $\overline{i}$ every time $m$ changes by, say, $1\%$, so that we can maintain the invariant that $\overline{i} \in \{\lfloor \log n \rfloor - 5, \lfloor \log n \rfloor - 4, \lfloor \log n \rfloor - 3\}$ at all times even as $\lfloor \log n\rceil$ changes over time.
    \item During each deletion, we spend $O(1)$ work clearing out the RAMs of each of $Y_{\lfloor \log n \rfloor - 2}, Y_{\lfloor \log n \rfloor - 1}$ (if they have any content). This ensures that, whenever a deletion causes the rightmost $Y_i$ to become $Z$, it is already an empty RAM. 
    \item On each insertion/deletion, we spend up to $O(1)$ work either building a new left-most $Y_i$ (i.e., rearranging a suffix of $X$ to become a new left-most $Y_i$, and creating the necessary metadata for the new $Y_i$) or eliminating the current left-most $Y_i$ (i.e., updating our metadata so that the left-most $Y_i$ can become part of $X$). Using hysteresis, this allows us to maintain the property that, at any given moment, $|X| \in \delta n \cdot (0.1, 0.9)$.
    \item On each insertion/deletion, we spend $O(1)$ time preparing the metadata so that, whenever a new empty $Z$ is created, we already have the metadata necessary (i.e., empty retrieval data structures and empty overflow linked lists) for it.
\end{itemize}

It is worth noting that the above migration work can straightforwardly be implemented while increasing the overall amount of metadata stored by at most an $O(1)$ factor. Thus, even with the migration work, the total amount of metadata stored in the core RAM is $\delta \cdot O(n \log n)$ bits.

\begin{remark}[Reducing to a Polynomial-Size Universe in Each Array Segment]
    Whereas fixed-capacity hash tables can assume a polynomial-size universe without loss of generality (by simply hashing down to a polynomial-size intermediate universe $U'$), dynamically-resized hash tables cannot. Since the $Y_i$'s and $Z$ are fixed-capacity data structures (technically $Z$ is a prefix of one), we can perform universe reduction within them exactly as we would for a fixed-capacity hash table. The first array segment $X$ is not a fixed-capacity data structure, but we can update the size of the intermediate universe $U'$ that we use for $X$ each time that we copy the metadata for $X$ from one $Y_{\overline{i}}$ to another $Y_{\overline{i}'}$. Thus, we can maintain the invariant that, at any given moment, each array segment treats the keys as being from a polynomial-size universe.
    \label{rem:universe}
\end{remark}

\paragraph{Supporting queries.}
To query an element $x$, we query all of $X$ (using the $X$'s retrieval data structure), $\{Y_i\}$, and $Z$. Each of these queries takes $O(1)$ time making for $O(\log \delta^{-1}) = O(1)$ total query time.

\paragraph{Supporting insertions.}
Now suppose we wish to insert an element $x$. The basic idea is that $x$ gets inserted into $Z$. If, after the insertion, $Z$ is small enough that it consists only of a backyard (and no bin slots), then the insertion is implemented by simply placing the element at the end of the array, and updating the metadata for $Z$ (which is stored not in $Z$ but in $Y_{\overline{i}}$) appropriately. Otherwise, the insertion is implemented as follows. First, we add a slot to the end of the array which, in turn, increments the size of some bin $j$. Using the overflow linked list for that bin $j$, we find an element in the backyard that can be moved into the bin (such an element exists with high probability by \cref{claim:resizeunderflow}). This creates a free slot in the backyard which is used for the newly inserted element $x$. Finally, all of the metadata for $Z$ is updated appropriately. 

Additionally, whenever an insertion is performed, we also spend $O(1)$ time on migration work, as described above. If, after the insertion, $Z$ is full (size $2^{\lfloor \log n \rfloor}$), then the insertion causes $Z$ to become the new right-most $Y_i$, and a new (empty) $Z$ is created.

\paragraph{Supporting deletions.}
Next, suppose we wish to delete an element $x$. To begin, suppose $x \in Z$. Then we first use the procedure for deleting from a fixed-size partner hash table in order to delete $x$ (this creates a free slot $k$ in $Z$'s backyard). We then take the element in slot $n$ (where $n$ is the size before the deletion) and move it to slot $k$, updating the metadata for the hash table appropriately. Notably, if, before the deletion, slot $n$ corresponded to a slot in bin $j$, then we remove that slot from the bin, and update the overflow linked list for bin $j$ to reflect the fact that an element has been evicted from the bin; we also update the retrieval data structures for both the frontyard and backyard of $Z$.

Now, moving to the case where $x$ is not in $Z$, because it is in either $X$ or some $Y_i$, we can reduce to the case of $x \in Z$ as follows. Let $Q$ be the segment of the array containing $x$ ($Q$ is either $X$ or some $Y_i$). We delete $x$ by selecting a random $z \in Z$, deleting $z$ from $Z$ (using the procedure above), deleting $x$ from $Q$ (using the procedure for deleting from a fixed-size partner hash table), and inserting $z$ into $Q$ (again using the procedure for a fixed-size partner hash table).

Finally, whenever a deletion is performed, we also spend $O(1)$ time on migration work. If, prior to the deletion, $Z$ was already empty, then the deletion causes the right-most $Y_i$ to become the new $Z$.

\paragraph{Putting the pieces together.}
By design, the total amount of metadata, at any given moment, is $\delta \cdot O(n \log n)$. On the other hand, the RAMs implemented by $Y_{\overline{i} - 1}, Y_{\overline{i}}, Y_{\overline{i} + 1}$ all have sizes $\Theta(n \log n)$ bits with a hidden constant independent of $\delta$. It follows that, so long as $\delta$ is taken to be a sufficiently small positive constant, we will never attempt to store more metadata in a RAM than can be fit.

The high-level structure of $X$, $\{Y_i\}$ and $Z$ is maintained by the $O(1)$ migration work that is performed on each insertion and deletion. Besides this work, each insertion and deletion spends $O(1)$ updating the data structure. The only possible failure modes are that (1) some retrieval data structure fails, which occurs with probability at most $1/\poly(n)$; (2) some bin in some $Y_i$ underflows, which by the analysis in Section \ref{sec:fixed_size} occurs with probability $1 / \poly(|Y_i|) = 1 / \poly(n)$; (3) some bin in $Z$ underflows, which by \cref{claim:resizeunderflow} occurs with probability $1 / \poly(n)$; or (4) some RAM write in some $Y_i$ fails, which by the analysis in Section \ref{sec:fixed_size}, occurs with probability $1 / \poly(|Y_i|) = 1 / \poly(n)$. Such failures can easily be detected, at which point the data structure can be rebuilt from scratch using new hash functions.

Barring $1 / \poly(n)$-probability failure events, each insertion/deletion is $O(1)$ time. Finally, queries take $O(\log \delta^{-1}) = O(1)$ worst-case time, since they spend $O(1)$ time on each segment of the array.

\bibliographystyle{alpha}
\bibliography{ref,hashing_massive}

\newcommand{\etalchar}[1]{$^{#1}$}
\begin{thebibliography}{DGM{\etalchar{+}}10}

\bibitem[AK74]{amble1974ordered}
Ole Amble and Donald~Ervin Knuth.
\newblock Ordered hash tables.
\newblock {\em The Computer Journal}, 17(2):135--142, January 1974.

\bibitem[ANS09]{arbitman2009amortized}
Yuriy Arbitman, Moni Naor, and Gil Segev.
\newblock De-amortized cuckoo hashing: Provable worst-case performance and experimental results.
\newblock In {\em Proc. 36th International Colloquium on Automata, Languages and Programming (ICALP)}, pages 107--118, 2009.

\bibitem[ANS10]{ArbitmanNaSe10}
Yuriy Arbitman, Moni Naor, and Gil Segev.
\newblock Backyard cuckoo hashing: Constant worst-case operations with a succinct representation.
\newblock In {\em Proc. 51st IEEE Symposium on Foundations of Computer Science (FOCS)}, pages 787--796, 2010.

\bibitem[Bat75]{batagelj1975quadratic}
Vladimir Batagelj.
\newblock The quadratic hash method when the table size is not a prime number.
\newblock {\em Communications of the ACM}, 18(4):216--217, 1975.

\bibitem[BCF{\etalchar{+}}23]{bender2023iceberg}
Michael~A. Bender, Alex Conway, Mart{\'\i}n {Farach-Colton}, William Kuszmaul, and Guido Tagliavini.
\newblock Iceberg hashing: Optimizing many hash-table criteria at once.
\newblock {\em Journal of the ACM}, 70(6):1--51, 2023.

\bibitem[BE23]{bercea2023dynamic}
Ioana~O. Bercea and Guy Even.
\newblock Dynamic dictionaries for multisets and counting filters with constant time operations.
\newblock {\em Algorithmica}, 85(6):1786--1804, 2023.

\bibitem[BF24]{bell20241}
Tolson Bell and Alan Frieze.
\newblock {$O(1)$} insertion for random walk $d$-ary cuckoo hashing up to the load threshold.
\newblock In {\em Proc. 65th IEEE Symposium on Foundations of Computer Science (FOCS)}, 2024.

\bibitem[BFK{\etalchar{+}}22]{bender2022optimal}
Michael~A. Bender, Mart{\'i}n {Farach-Colton}, John Kuszmaul, William Kuszmaul, and Mingmou Liu.
\newblock On the optimal time/space tradeoff for hash tables.
\newblock In {\em Proc.\ 54th ACM SIGACT Symposium on Theory of Computing (STOC)}, pages 1284--1297, 2022.

\bibitem[BFKK24]{bender2024modern}
Michael~A. Bender, Mart{\'i}n {Farach-Colton}, John Kuszmaul, and William Kuszmaul.
\newblock Modern hashing made simple.
\newblock In {\em Proc. 7th Symposium on Simplicity in Algorithms (SOSA)}, pages 363--373. SIAM, 2024.

\bibitem[BG07]{DBLP:conf/focs/BlellochG07}
Guy~E. Blelloch and Daniel Golovin.
\newblock Strongly history-independent hashing with applications.
\newblock In {\em Proc. 48th {IEEE} Symposium on Foundations of Computer Science (FOCS)}, pages 272--282. {IEEE}, 2007.

\bibitem[BKK22]{bender2022linear}
Michael~A. Bender, Bradley~C. Kuszmaul, and William Kuszmaul.
\newblock Linear probing revisited: Tombstones mark the demise of primary clustering.
\newblock In {\em Proc. 62nd IEEE Symposium on Foundations of Computer Science (FOCS)}, pages 1171--1182, 2022.

\bibitem[BKZ24]{oblivious}
Michael~A. Bender, William Kuszmaul, and Renfei Zhou.
\newblock Tight bounds for classical open addressing.
\newblock In {\em Proc.\ 65th IEEE Symposium on Foundations of Computer Science (FOCS)}, 2024.

\bibitem[Bre73]{brent1973reducing}
Richard~P. Brent.
\newblock Reducing the retrieval time of scatter storage techniques.
\newblock {\em Communications of the ACM}, 16(2):105--109, 1973.

\bibitem[Bur05]{burkhard2005external}
Walter~A. Burkhard.
\newblock External double hashing with choice.
\newblock In {\em Proc. 8th International Symposium on Parallel Architectures, Algorithms and Networks (ISPAN)}, pages 100--107. IEEE, 2005.

\bibitem[CFS18]{ConwayFaSh18}
Alexander Conway, Mart{\'\i}n Farach{-}Colton, and Philip Shilane.
\newblock Optimal hashing in external memory.
\newblock In {\em Proc. 45th International Colloquium on Automata, Languages, and Programming (ICALP)}, pages 39:1--39:14, 2018.

\bibitem[CK09]{cohen2009bounds}
Jeffrey~S. Cohen and Daniel~M. Kane.
\newblock Bounds on the independence required for cuckoo hashing.
\newblock {\em ACM Transactions on Algorithms}, 2009.

\bibitem[CLM85]{celis1985robin}
Pedro Celis, Per-Ake Larson, and J.~Ian Munro.
\newblock Robin {Hood} hashing.
\newblock In {\em Proc. 26th IEEE Symposium on Foundations of Computer Science (SFCS)}, pages 281--288. IEEE, 1985.

\bibitem[CSW07]{cain2007random}
Julie~Anne Cain, Peter Sanders, and Nick Wormald.
\newblock The random graph threshold for $k$-orientiability and a fast algorithm for optimal multiple-choice allocation.
\newblock In {\em Proc. 18th ACM-SIAM Symposium on Discrete Algorithms (SODA)}, pages 469--476, 2007.

\bibitem[CW77]{carter1977universal}
J.~Lawrence Carter and Mark~N. Wegman.
\newblock Universal classes of hash functions.
\newblock In {\em Proc. 9th ACM Symposium on Theory of Computing (STOC)}, pages 106--112, 1977.

\bibitem[DGM{\etalchar{+}}10]{DietzfelbingerGoMi10}
Martin Dietzfelbinger, Andreas Goerdt, Michael Mitzenmacher, Andrea Montanari, Rasmus Pagh, and Michael Rink.
\newblock Tight thresholds for cuckoo hashing via {XORSAT}.
\newblock In {\em Proc. 37th International Colloquium on Automata, Languages and Programming (ICALP)}, pages 213--225, 2010.

\bibitem[DM90]{dietzfelbinger1990new}
Martin Dietzfelbinger and Friedhelm {Meyer auf der Heide}.
\newblock A new universal class of hash functions and dynamic hashing in real time.
\newblock In {\em Proc. 17th International Colloquium on Automata, Languages, and Programming (ICALP)}, pages 6--19. Springer, 1990.

\bibitem[DMPP06]{demaine2006dictionariis}
Erik~D. Demaine, Friedhelm {Meyer auf der Heide}, Rasmus Pagh, and Mihai P{\v a}tra{\c s}cu.
\newblock De dictionariis dynamicis pauco spatio utentibus.
\newblock In {\em Proc.\ 7th Latin American Conference on Theoretical Informatics (LATIN)}, pages 349--361, 2006.

\bibitem[DR09]{dietzfelbinger2009applications}
Martin Dietzfelbinger and Michael Rink.
\newblock Applications of a splitting trick.
\newblock In {\em Proc. 36th International Colloquium on Automata, Languages, and Programming (ICALP) Part I}, pages 354--365. Springer, 2009.

\bibitem[DW07]{dietzfelbinger2007balanced}
Martin Dietzfelbinger and Christoph Weidling.
\newblock Balanced allocation and dictionaries with tightly packed constant size bins.
\newblock {\em Theoretical Computer Science}, 380(1-2):47--68, 2007.

\bibitem[Eck74]{ecker1974period}
A.~Ecker.
\newblock The period of search for the quadratic and related hash methods.
\newblock {\em The Computer Journal}, 17(4):340--343, 1974.

\bibitem[FG03a]{franceschini2003implicit}
Gianni Franceschini and Roberto Grossi.
\newblock Implicit dictionaries supporting searches and amortized updates in {$O(\log n \log \log n)$} time.
\newblock In {\em Proc. 14th ACM-SIAM Symposium on Discrete Algorithms (SODA)}, pages 670--678, 2003.

\bibitem[FG03b]{franceschini2003optimal}
Gianni Franceschini and Roberto Grossi.
\newblock Optimal cache-oblivious implicit dictionaries.
\newblock In {\em Proc. 30th International Conference on Automata, Languages, and Programming (ICALP)}, pages 316--331, 2003.

\bibitem[FG06]{franceschini2006optimal}
Gianni Franceschini and Roberto Grossi.
\newblock Optimal implicit dictionaries over unbounded universes.
\newblock {\em Theory of Computing Systems}, 39(2):321--345, April 2006.

\bibitem[FGMP02]{franceschini2002implicita}
Gianni Franceschini, Roberto Grossi, J.~Ian Munro, and Linda Pagli.
\newblock Implicit {B-trees}: New results for the dictionary problem.
\newblock In {\em Proc. 43rd IEEE Symposium on Foundations of Computer Science (FOCS)}, pages 145--154, 2002.

\bibitem[FJ19]{frieze2019insertion}
Alan Frieze and Tony Johansson.
\newblock On the insertion time of random walk cuckoo hashing.
\newblock {\em Random Structures \& Algorithms}, 54(4):721--729, 2019.

\bibitem[FKK24]{greedy}
Mart{\'i}n {Farach-Colton}, Andrew Krapivin, and William Kuszmaul.
\newblock Tight bounds for open addressing without reordering.
\newblock In {\em Proc.\ 65th IEEE Symposium on Foundations of Computer Science (FOCS)}, 2024.

\bibitem[FKP16]{fountoulakis2016multiple}
Nikolaos Fountoulakis, Megha Khosla, and Konstantinos Panagiotou.
\newblock The multiple-orientability thresholds for random hypergraphs.
\newblock {\em Combinatorics, Probability and Computing}, 25(6):870--908, 2016.

\bibitem[FM12]{frieze2012maximum}
Alan Frieze and P{\'a}ll Melsted.
\newblock Maximum matchings in random bipartite graphs and the space utilization of cuckoo hash tables.
\newblock {\em Random Structures \& Algorithms}, 41(3):334--364, 2012.

\bibitem[FMM11]{frieze2011analysis}
Alan Frieze, P{\'a}ll Melsted, and Michael Mitzenmacher.
\newblock An analysis of random-walk cuckoo hashing.
\newblock {\em SIAM Journal on Computing}, 40(2):291--308, 2011.

\bibitem[FN93]{fiat1993implicit}
Amos Fiat and Moni Naor.
\newblock Implicit {$O(1)$} probe search.
\newblock {\em SIAM Journal on Computing}, 22(1):1--10, 1993.

\bibitem[FNSS92]{fiat1988nonoblivious}
Amos Fiat, Moni Naor, Jeanette~P. Schmidt, and Alan Siegel.
\newblock Nonoblivious hashing.
\newblock {\em Journal of the ACM}, 39(4):764--782, October 1992.
\newblock See also STOC 1988.

\bibitem[FP10]{fountoulakis2010orientability}
Nikolaos Fountoulakis and Konstantinos Panagiotou.
\newblock Orientability of random hypergraphs and the power of multiple choices.
\newblock In {\em Proc. 37th International Colloquium on Automata, Languages, and Programming (ICALP)}, pages 348--359. Springer, 2010.

\bibitem[FP18]{frieze2018balanced}
Alan Frieze and Samantha Petti.
\newblock Balanced allocation through random walk.
\newblock {\em Information Processing Letters}, 131:39--43, 2018.

\bibitem[FPS13]{fountoulakis2013insertion}
Nikolaos Fountoulakis, Konstantinos Panagiotou, and Angelika Steger.
\newblock On the insertion time of cuckoo hashing.
\newblock {\em SIAM Journal on Computing}, 42(6):2156--2181, 2013.

\bibitem[FR07]{fernholz2007k}
Daniel Fernholz and Vijaya Ramachandran.
\newblock The $k$-orientability thresholds for {$G_{n, p}$}.
\newblock In {\em Proc. 18th ACM-SIAM Symposium on Discrete Algorithms (SODA)}, pages 459--468, 2007.

\bibitem[Fre83]{frederickson1983implicit}
Greg~N. Frederickson.
\newblock Implicit data structures for the dictionary problem.
\newblock {\em Journal of the ACM}, 30(1):80--94, January 1983.

\bibitem[GHMT12]{GoodrichHiMi12}
Michael~T. Goodrich, Daniel~S. Hirschberg, Michael Mitzenmacher, and Justin Thaler.
\newblock Cache-oblivious dictionaries and multimaps with negligible failure probability.
\newblock In {\em Proc. 1st Mediterranean Conference on Algorithms (MedAlg)}, pages 203--218, 2012.

\bibitem[GM79]{gonnet1979efficient}
Gaston~H. Gonnet and J.~Ian Munro.
\newblock Efficient ordering of hash tables.
\newblock {\em SIAM Journal on Computing}, 8(3):463--478, 1979.

\bibitem[Gon77]{gonnetinterpolation}
Gaston~H. Gonnet.
\newblock Interpolation and interpolation hash searching.
\newblock Research Report CS-77-02, Department of Computer Science, University of Waterloo, January 1977.

\bibitem[GS78]{GuibasSz78}
Leo~J. Guibas and Endre Szemer{\'{e}}di.
\newblock The analysis of double hashing.
\newblock {\em Journal of Computer and System Sciences}, 16(2):226--274, April 1978.

\bibitem[HD72]{HopgoodDa72}
F.~Robert~A. Hopgood and James~H. Davenport.
\newblock The quadratic hash method when the table size is a power of 2.
\newblock {\em The Computer Journal}, 15(4):314--315, 1972.

\bibitem[HMP01]{HagerupMiPa01}
Torben Hagerup, Peter~Bro Miltersen, and Rasmus Pagh.
\newblock Deterministic dictionaries.
\newblock {\em Journal of Algorithms}, 41(1):69--85, October 2001.

\bibitem[IP12]{IaconoPa12}
John Iacono and Mihai P\v{a}tra\c{s}cu.
\newblock Using hashing to solve the dictionary problem.
\newblock In {\em Proc. 23rd ACM-SIAM Symposium on Discrete Algorithms (SODA)}, pages 570--582, 2012.

\bibitem[JP08]{JensenPa08}
Morten~Skaarup Jensen and Rasmus Pagh.
\newblock Optimality in external memory hashing.
\newblock {\em Algorithmica}, 52(3):403--411, November 2008.

\bibitem[KA19]{khosla2019faster}
Megha Khosla and Avishek Anand.
\newblock A faster algorithm for cuckoo insertion and bipartite matching in large graphs.
\newblock {\em Algorithmica}, 81(9):3707--3724, 2019.

\bibitem[KM07]{kirsch2007using}
Adam Kirsch and Michael Mitzenmacher.
\newblock Using a queue to de-amortize cuckoo hashing in hardware.
\newblock In {\em Proc. 45th Allerton Conference on Communication, Control, and Computing (Allerton)}, volume~75, 2007.

\bibitem[KM25]{kuszmaul2025efficient}
William Kuszmaul and Michael Mitzenmacher.
\newblock Efficient $d$-ary cuckoo hashing at high load factors by bubbling up.
\newblock In {\em Proc. 36th ACM-SIAM Symposium on Discrete Algorithms (SODA)}, pages 3931--3952. SIAM, 2025.

\bibitem[KMW10]{kirsch2010more}
Adam Kirsch, Michael Mitzenmacher, and Udi Wieder.
\newblock More robust hashing: Cuckoo hashing with a stash.
\newblock {\em SIAM Journal on Computing}, 39(4):1543--1561, 2010.

\bibitem[Knu63]{knuth1963notes}
Donald~E. Knuth.
\newblock Notes on ``open'' addressing.
\newblock Available online at \url{https://jeffe.cs.illinois.edu/teaching/datastructures/2011/notes/knuth-OALP.pdf}, 1963.

\bibitem[Knu98]{Knuth98Vol3}
Donald~E. Knuth.
\newblock {\em The Art of Computer Programming, Volume III: Sorting and Searching}.
\newblock Addison-Wesley, 2nd edition, 1998.

\bibitem[Kus22]{kuszmaul2022hash}
William Kuszmaul.
\newblock A hash table without hash functions, and how to get the most out of your random bits.
\newblock In {\em Proc. 63rd IEEE Symposium on Foundations of Computer Science (FOCS)}, pages 991--1001, 2022.

\bibitem[Kus23]{kuszmaul2023strongly}
William Kuszmaul.
\newblock Strongly history-independent storage allocation: New upper and lower bounds.
\newblock In {\em Proc.\ 64th IEEE Symposium on Foundations of Computer Science (FOCS)}, pages 1822--1841, 2023.

\bibitem[KW66]{konheim1966occupancy}
Alan~G. Konheim and Benjamin Weiss.
\newblock An occupancy discipline and applications.
\newblock {\em SIAM Journal on Applied Mathematics}, 14(6):1266--1274, November 1966.

\bibitem[KX24]{kuszmaul2024towards}
William Kuszmaul and Zoe Xi.
\newblock Towards an analysis of quadratic probing.
\newblock In {\em Proc. 51st International Colloquium on Automata, Languages, and Programming (ICALP)}. Schloss Dagstuhl--Leibniz-Zentrum f{\"u}r Informatik, 2024.

\bibitem[Lel12]{lelarge2012new}
Marc Lelarge.
\newblock A new approach to the orientation of random hypergraphs.
\newblock In {\em Proc. 23rd ACM-SIAM Symposium on Discrete Algorithms (SODA)}, pages 251--264, 2012.

\bibitem[LLYZ23a]{li2023dynamic}
Tianxiao Li, Jingxun Liang, Huacheng Yu, and Renfei Zhou.
\newblock Dynamic ``succincter''.
\newblock In {\em Proc.\ 64th IEEE Symposium on Foundations of Computer Science (FOCS)}, pages 1715--1733, 2023.

\bibitem[LLYZ23b]{li2023tight}
Tianxiao Li, Jingxun Liang, Huacheng Yu, and Renfei Zhou.
\newblock Tight cell-probe lower bounds for dynamic succinct dictionaries.
\newblock In {\em Proc.\ 64th IEEE Symposium on Foundations of Computer Science (FOCS)}, pages 1842--1862, 2023.

\bibitem[LLYZ24]{li2024dynamic}
Tianxiao Li, Jingxun Liang, Huacheng Yu, and Renfei Zhou.
\newblock Dynamic dictionary with subconstant wasted bits per key.
\newblock In {\em Proc.\ 35th ACM-SIAM Symposium on Discrete Algorithms (SODA)}, pages 171--207, 2024.

\bibitem[LM88]{lueker1988more}
George Lueker and Mariko Molodowitch.
\newblock More analysis of double hashing.
\newblock In {\em Proc. 20th ACM Symposium on Theory of Computing (STOC)}, pages 354--359, 1988.

\bibitem[LP09]{lehman20093}
Eric Lehman and Rina Panigrahy.
\newblock 3.5-way cuckoo hashing for the price of 2-and-a-bit.
\newblock In {\em 17th Annual European Symposium on Algorithms}, pages 671--681, 2009.

\bibitem[Mal77]{mallach1977scatter}
Efrem~G. Mallach.
\newblock Scatter storage techniques: A unifying viewpoint and a method for reducing retrieval times.
\newblock {\em The Computer Journal}, 20(2):137--140, 1977.

\bibitem[Mau83]{Maurer68}
Ward~Douglas Maurer.
\newblock An improved hash code for scatter storage.
\newblock {\em Communications of the ACM}, 26(1):36--38, January 1983.

\bibitem[MB03]{martini2003double}
Paul~M. Martini and Walter~A. Burkhard.
\newblock Double hashing with multiple passbits.
\newblock {\em International Journal of Foundations of Computer Science}, 14(06):1165--1182, 2003.

\bibitem[MC86]{munro1986techniques}
J.~Ian Munro and Pedro Celis.
\newblock Techniques for collision resolution in hash tables with open addressing.
\newblock In {\em Proc. 1986 ACM Fall joint computer conference}, pages 601--610, 1986.

\bibitem[Mol90]{molodowitch1990analysis}
Mariko Molodowitch.
\newblock {\em Analysis and design of algorithms: Double hashing and parallel graph searching}.
\newblock University of California, Irvine, 1990.

\bibitem[MS80]{munro1980implicit}
Jaeson Munro and Hendra Suwanda.
\newblock Implicit data structures for fast search and update.
\newblock {\em Journal of Computer and System Sciences}, 21:236--250, 1980.

\bibitem[MT12]{mitzenmacher2012peeling}
Michael Mitzenmacher and Justin Thaler.
\newblock Peeling arguments and double hashing.
\newblock In {\em Proc. 50th Allerton Conference on Communication, Control, and Computing (Allerton)}, pages 1118--1125. IEEE, 2012.

\bibitem[Mun84]{munro1984implicita}
J.~Ian Munro.
\newblock An implicit data structure for the dictionary problem that runs in polylog time.
\newblock In {\em Proc. 25th IEEE Symposium on Foundations of Computer Science (FOCS)}, pages 369--374, 1984.

\bibitem[Mun86]{munro1986implicit}
J.~Ian Munro.
\newblock An implicit data structure supporting insertion, deletion, and search in {$O(\log^2 n)$} time.
\newblock {\em Journal of Computer and System Sciences}, 33(1):66--74, 1986.

\bibitem[NT01]{DBLP:journals/iacr/NaorT01}
Moni Naor and Vanessa Teague.
\newblock Anti-persistence: History independent data structures.
\newblock In {\em Proc. 33rd ACM Symposium on Theory of Computing (STOC)}, pages 492--501, 2001.

\bibitem[OP03]{ostlin2003uniform}
Anna Ostlin and Rasmus Pagh.
\newblock Uniform hashing in constant time and linear space.
\newblock In {\em Proc. 35th ACM Symposium on Theory of Computing (STOC)}, pages 622--628, 2003.

\bibitem[PPR09]{pagh2009linear}
Anna Pagh, Rasmus Pagh, and Milan Ru{\v z}i{\'c}.
\newblock Linear probing with constant independence.
\newblock {\em SIAM Journal on Computing}, 39(3):1107--1120, January 2009.
\newblock See also STOC 2007.

\bibitem[PR01]{pagh2001cuckoo}
Rasmus Pagh and Flemming~Friche Rodler.
\newblock Cuckoo hashing.
\newblock In {\em Proc.\ 9th European Symposium on Algorithms (ESA)}, pages 121--133, 2001.

\bibitem[PT14]{PatrascuTh14}
Mihai P{\v a}tra{\c{s}}cu and Mikkel Thorup.
\newblock Dynamic integer sets with optimal rank, select, and predecessor search.
\newblock In {\em Proc. 55th IEEE Symposium on Foundations of Computer Science (FOCS)}, pages 166--175, Philadelphia, Pennsylvania, USA, 18--21~October 2014.

\bibitem[PT16]{patrascu2016independence}
Mihai P{\v a}tra{\c s}cu and Mikkel Thorup.
\newblock On the {$k$}-independence required by linear probing and minwise independence.
\newblock {\em ACM Transactions on Algorithms}, 12(1):1--27, February 2016.
\newblock See also ICALP 2010.

\bibitem[Rad70]{radke1970use}
Charles~E. Radke.
\newblock The use of quadratic residue research.
\newblock {\em Communications of the ACM}, 13(2):103--105, 1970.

\bibitem[RS03]{Raman03Succinct}
Rajeev Raman and Srinivasa~Rao Satti.
\newblock Succinct dynamic dictionaries and trees.
\newblock In {\em Proc. 30th International Colloquium on Automata, Languages and Programming (ICALP)}, pages 357--368, 2003.

\bibitem[Ru{\v{z}}08]{Ruzic08}
Milan Ru{\v{z}}i{\'{c}}.
\newblock Uniform deterministic dictionaries.
\newblock {\em ACM Transactions on Algorithms}, 4(1):1--23, March 2008.

\bibitem[Sie89]{siegel1989universal}
Alan Siegel.
\newblock On universal classes of fast high performance hash functions, their time-space tradeoff, and their applications.
\newblock In {\em Proc. 30th IEEE Symposium on Foundations of Computer Science}, pages 20--25. IEEE Computer Society, 1989.

\bibitem[Sie95]{siegel1995universal}
Alan Siegel.
\newblock On universal classes of extremely random constant time hash functions and their time-space tradeoff.
\newblock Technical report, Technical Report TR1995-684, New York University, 1995.

\bibitem[Sie04]{siegel2004universal}
Alan Siegel.
\newblock On universal classes of extremely random constant-time hash functions.
\newblock {\em SIAM Journal on Computing}, 33(3):505--543, 2004.

\bibitem[SSS95]{schmidt1995chernoffhoeffding}
Jeanette~P. Schmidt, Alan Siegel, and Aravind Srinivasan.
\newblock {Chernoff-Hoeffding} bounds for applications with limited independence.
\newblock {\em SIAM Journal on Discrete Mathematics}, 8(2):223--250, May 1995.

\bibitem[Sun91]{Sundar91}
Rajamani Sundar.
\newblock A lower bound for the dictionary problem under a hashing model.
\newblock In {\em Proc. 32nd Symposium of Foundations of Computer Science (FOCS)}, pages 612--621, San Juan, Puerto Rico, USA, October 1991.

\bibitem[Ull72]{ullman1972note}
Jeffrey~D. Ullman.
\newblock A note on the efficiency of hashing functions.
\newblock {\em Journal of the ACM (JACM)}, 19(3):569--575, 1972.

\bibitem[VZ13]{verbin2013limits}
Elad Verbin and Qin Zhang.
\newblock The limits of buffering: A tight lower bound for dynamic membership in the external memory model.
\newblock {\em SIAM Journal on Computing}, 42(1):212--229, 2013.

\bibitem[Wal22]{walzer2022insertion}
Stefan Walzer.
\newblock Insertion time of random walk cuckoo hashing below the peeling threshold.
\newblock In {\em Proc. 30th European Symposium on Algorithms (ESA)}, pages 87:1--87:11, 2022.

\bibitem[Wal23]{walzer2023load}
Stefan Walzer.
\newblock Load thresholds for cuckoo hashing with overlapping blocks.
\newblock {\em ACM Transactions on Algorithms}, 19(3):1--22, 2023.

\bibitem[WC79]{wegman1979new}
Mark~N. Wegman and J.~Lawrence Carter.
\newblock New classes and applications of hash functions.
\newblock In {\em Proc. 20th IEEE Symposium on Foundations of Computer Science (FOCS)}, pages 175--182. IEEE, 1979.

\bibitem[Yao81]{yao1981should}
Andrew Chi-Chih Yao.
\newblock Should tables be sorted?
\newblock {\em Journal of the ACM}, 28(3):615--628, July 1981.

\bibitem[Yao85]{yao1985uniform}
Andrew Chi-Chih Yao.
\newblock Uniform hashing is optimal.
\newblock {\em Journal of the ACM}, 32(3):687--693, July 1985.

\end{thebibliography}

\end{document}